\documentclass[letterpaper, 11pt, notitlepage]{article}

\usepackage[margin=1in, pass]{geometry}
\usepackage{hyperref}
\usepackage{amsmath}
\usepackage{amsthm}
\usepackage{enumerate}
\usepackage{fullpage}
\usepackage{appendix}
\usepackage{multicol}
\usepackage{xcolor}
\usepackage{tikz}
\usetikzlibrary{shapes,snakes}

\usepackage{algorithm}
\floatname{algorithm}{Procedure}
\usepackage{algorithmicx}

\usepackage{tabulary}

\usepackage{float}

\newcommand{\remove}[1]{}

\usepackage[noend]{algpseudocode}
\usepackage{amssymb}
\usepackage[colorinlistoftodos,prependcaption,textsize=small]{todonotes}
\usepackage{verbatim}
\usepackage{color}
\usepackage{tabulary}
\usepackage{mathtools}
\DeclarePairedDelimiter\floor{\lfloor}{\rfloor}
\usepackage{tcolorbox}
\newcommand\epsi{\epsilon^{-i}}

\newcommand\frho{{\rho^{-1}}}
\newcommand\eps[1]{\epsilon^{-{#1}}}

\newcommand\nfrac{n^{1+{1}/{\kappa}}}
\newcommand\congest{{\sf CONGEST}}
\newcommand\local{{\sf LOCAL}}
\newcommand\congestmo{{\sf CONGEST} model}
\newcommand\localmo{{\sf LOCAL} model}

\newcommand\bbeta{\left(
	\frac{O\left(\log\kappa\rho +\rho^{-1}\right)}{\rho\cdot \epsilon}
	\right)^{\log\kappa\rho +\rho^{-1}+O(1)}}

\makeatletter
\newcommand{\removelatexerror}{\let\@latex@error\@gobble}
\makeatother
\usepackage{wrapfig}

\usepackage{cleveref}
\usepackage{adjustbox}

\newtheorem{claim-subsection}{Claim}[subsection]
\newtheorem{theorem}{Theorem}[section]
\newtheorem{corollary}[theorem]{Corollary}
\newtheorem{lemma}[theorem]{Lemma}

\usepackage{graphicx}

\usepackage{titlesec}
\titleformat{\subsection}[runin]
{\normalfont\large\bfseries}{\thesubsection}{1em}{}
\titleformat{\subsubsection}[runin]
{\normalfont\normalsize\bfseries}{\thesubsubsection}{1em}{}

\titleformat{\paragraph}
{\normalfont\normalsize\bfseries}{\theparagraph}{1em}{}
\titlespacing*{\paragraph}
{0pt}{3.25ex plus 1ex minus .2ex}{1.5ex plus .2ex}

\title{Near-Additive Spanners In Low Polynomial Deterministic CONGEST Time}
\author{Michael Elkin$^1$ and Shaked Matar$^1$}
\date{$^1$Department of Computer Science, Ben-Gurion University of the Negev, Beer-Sheva, Israel.\\
Email: \texttt{elkinm@cs.bgu.ac.il, matars@post.bgu.ac.il}}
\begin{document}
\maketitle

\begin{abstract}
	Given a pair of parameters $\alpha\geq 1,\beta\geq 0$, a subgraph $G'=(V,H)$ of an $n$-vertex unweighted undirected graph $G=(V,E)$ is called an
	\textit{$(\alpha,\beta)$-spanner} if for every pair $u,v\in V$ of vertices, we have $d_{G'}(u,v)\leq \alpha d_{G}(u,v)+\beta$. If $\beta=0$ the spanner
	is called a \textit{multiplicative} $\alpha$-spanner, and if $\alpha = 1+\epsilon$, for an arbitrarily small $\epsilon>0$, the spanner is said to be a
	\textit{near-additive} one.
	
	Graph spanners \cite{Awerbuch85,PelegS89} are a fundamental and extremely well-studied combinatorial construct, with a multitude of applications in 
	distributed computing and in other areas. Near-additive spanners, introduced in \cite{ElkinP01}, preserve large distances much more faithfully than the more 
	traditional multiplicative spanners. Also, recent lower bounds \cite{AbboudB15} ruled out the existence of arbitrarily sparse purely additive spanners 
	(i.e., spanners with $\alpha=1$), and therefore showed that essentially near-additive spanners provide the best approximation of distances that one 
	can hope for. 
	
	Numerous distributed algorithms, for constructing sparse near-additive spanners, were devised in \cite{Elkin01, ElkinZ06, DerbelGPV09, Pettie10,
		ElkinN17}. In particular, there are now known efficient randomized algorithms in the \congestmo\ that construct such spanners \cite{ElkinN17},
	and also there are efficient deterministic algorithms in the \localmo\ \cite{DerbelGPV09}. However, the only known deterministic \congest-model 
	algorithm for the problem \cite{Elkin01} requires super-linear time in $n$. In this paper we remedy the situation and devise an efficient 
	deterministic \congest-model algorithm for constructing arbitrarily sparse near-additive spanners. 
	
	The running time of our algorithm is \textit{low polynomial}, i.e., roughly $O(\beta \cdot n^\rho)$, where $\rho > 0$ is an arbitrarily small positive constant that affects the
	additive term $\beta$. In general, the parameters of our new algorithm and of the resulting spanner are at the same ballpark as the respective parameters of  the state-of-the-art
	randomized algorithm for the problem due to \cite{ElkinN17}.
\end{abstract}

\section{Introduction}

Graph \textit{spanners} were introduced in the context of distributed algorithms by Awerbuch, Peleg and their co-authors \cite{Awerbuch85,PelegU87_first_spanner,PelegS89,PelegU89,AwerbuchP90a} in the end of the  eighties, and were extensively studied since then. See, e.g., \cite{DerbelGPV09,GhaffariK18,GrossmanP17,AbboudBP17,AbboudB15,ElkinN17,ElkinP01,baswana2007simple,BaswanaKMP10}, and the references therein. Numerous applications of spanners in Distributed Algorithms and other related areas were discovered in 
\cite{Awerbuch85,PelegU87_first_spanner,AingworthCIM99,Shis99,DorHZ00,ElkinP01,Elkin05,ElkinZ06,HenzingerKN16,ThorupZ01_distance_oracles,RodittyTZ05_distance_oracles}.
Also, spanners are a fundamental combinatorial construct, and their existential properties are being extensively explored  \cite{AlthoferDDJS93,PelegS89,ElkinP01,BaswanaKMP10,ThorupZ06,Pettie09,AbboudB15,ElkinNS15,AbboudBP17}. 

In this paper, we focus on \textit{unweighted}, undirected graphs. Given a graph $G=(V,E)$, a subgraph $G'=(V,H)$, $H\subseteq E$, is said to be a (multiplicative) \textit{$t$-spanner} of $G$, if for every pair $u,v\in V$ of vertices, we have $d_{G'}(u,v)\leq t\cdot d_{G}(u,v)$, for a parameter $t\geq 1$, where $d_G$ (respectively, $d_{G'}$) stands for the distance in the graph $G$ (respectively, in the subgraph $G'$). 

A fundamental tradeoff for multiplicative spanners is that for any $\kappa=1,2,\dots$, and any $n$-vertex graph $G= (V,E)$, there exists a $(2\kappa-1)$-spanner with $O(\nfrac)$ edges \cite{AlthoferDDJS93,PelegS89}. Assuming Erd\"os-Simonovits girth conjecture (see e.g., \cite{ErdosS82}), this tradeoff is optimal. 
Efficient distributed algorithms for constructing multiplicative spanners that (nearly) realize this tradeoff were given in \cite{AwerbuchBCP93,Cohen98,baswana2007simple,Elkin07,derbelGPV08,ElkinN17,DerbelMZ10,GrossmanP17,GhaffariK18}.

A different variety of spanners, called \textit{near-additive spanners}, was presented by Elkin and Peleg in \cite{ElkinP01}. They showed that for every $\epsilon>0$ and $\kappa=1,2,\dots$, there exists $\beta=\beta(\kappa,\epsilon)$, such that for every $n$-vertex graph $G=(V,E)$ there exists a $(1+\epsilon,\beta)$-spanner $G'=(V,E)$, $H\subseteq E$, with $O_{\epsilon,\kappa}(\nfrac)$ edges. (The subgraph $G'$ is called a \textit{$(1+\epsilon,\beta)$-spanner} of $G$ if for every pair $u,v\in V$ of vertices, it holds that $d_{G'}(u,v)\leq (1+\epsilon)d_G(u,v)+\beta$.)

In \cite{ElkinP01}, the additive error $\beta$ is given by $\beta_{EP}=\left(\frac{{\log \kappa}}{\epsilon}\right)^{{\log \kappa}-1}$, and it remains the state-of-the-art. These spanners approximate large distances much more faithfully than multiplicative spanners. Recent lower bounds of Abboud and Bodwin \cite{AbboudB15}, based on previous results of Coppersmith and Elkin \cite{CoppersmithE05}, showed that this type of spanners is essentially the best one can hope for, as in general arbitrarily sparse purely additive spanners (i.e., with $\epsilon=0$) do not exist.
Abboud et al. \cite{AbboudBP17} showed that the dependence 
on $\epsilon$ in \cite{ElkinP01} is nearly tight. Specifically, they showed a lower bound of $\beta=\Omega\left(\frac{1}{\epsilon{\log k}}\right)^{{\log \kappa}-1}$.

Near-additive spanners were extensively studied in the past two decades \cite{ElkinP01,Elkin05,ElkinZ06,ThorupZ06,DerbelGP07,DerbelGPV09,Pettie09,Pettie10,AbboudB15,AbboudBP17,ElkinN17}, both in the context of distributed algorithms and in other settings. (In particular, they were used in the streaming setting \cite{ElkinZ06,ElkinN17}, and for dynamic algorithms \cite{BernsteinR11,HenzingerKN16}.) 
Efficient distributed deterministic constructions of such spanners using large messages (the so-called \localmo; see Section \ref{sec Computational Model} for relevant definitions) were given by Derbel et al. \cite{DerbelGPV09}. They devised two constructions. In the first one, the running time is $O(\beta)$, and the additive term $\beta$ is is roughly $\eps{(\kappa-2)}$, i.e., exponentially larger than $\beta_{EP}$. In the second construction of \cite{DerbelGPV09}, the running time is $O(\beta\cdot 2^{O(\sqrt{{\log n}})})$, and $\beta$ is roughly the same as in \cite{ElkinP01}, i.e., $\beta\approx \beta_{EP}$.

Derbel et al. \cite{DerbelGPV09} explicitly raised the question of bulding $(1+\epsilon,\beta)$-spanners in the distributed \congestmo\ (i.e., using short messages; see Section \ref{sec Computational Model} for the definition), and this question remained open since their work. Efficient distributed randomized algorithms (in the \congestmo) for constructing near-additive spanners were given in \cite{ElkinZ06,Pettie10,ElkinN17}. The state-of-the-art algorithm is due to \cite{ElkinN17}. They devised a randomized algorithm that constructs $(1+\epsilon,\beta)$-spanners with $O_{\epsilon,\rho,\kappa}(\nfrac)$ edges, with a \textit{low polynomial} running time, that is, running time ${O}_{\epsilon,\kappa,\rho}(n^\rho)$, for arbitrarily small $\epsilon,\rho,{1}/{\kappa}>0$, where $\beta_{EN}= \beta(\epsilon,\rho,\kappa)=O\left(
\frac{{\log \kappa\rho}+\rho^{-1}}{\epsilon}
\right)^{{\log \kappa\rho} +\rho^{-1}+ O(1)}$. Note that the additive term $\beta_{EN}$ of \cite{ElkinN17} is at the same ballpark as the existential bound $\beta_{EP}$ of \cite{ElkinP01}.

The importance of devising deterministic distributed algorithms that work in the \congestmo\ for constructing spanners was recently articulated by Barenboim et al. \cite{BarenboimEG15}, Grossman and Parter \cite{GrossmanP17} and Ghaffari and Kuhn \cite{GhaffariK18}. All these authors devised algorithms for constructing multiplicative spanners. The only deterministic \congest-model algorithm for building $(1+\epsilon,\beta)$-spanners with $\tilde{O}_{\epsilon,\kappa,\rho}(\nfrac)$ edges is due to \cite{Elkin05}. That algorithm suffers however from a superlinear in $n$ running time, $O(n^{1+\frac{1}{2\kappa}})$, and from additive term
$\beta_E=\left(\frac{\kappa}{\epsilon}\right)^{{\log \kappa}}\cdot (\rho^{-1})^{\rho^{-1}}$, which is significantly larger than $\beta_{EN}$ (the additive term of \cite{ElkinN17}). Also, the number of edges in the resulting spanner is by a $polylog(n)$ factor larger than the desired bound of $O_{\epsilon,\kappa,\rho}(\nfrac)$ (see \cite{ElkinP01,Pettie10,ElkinN17,AbboudBP17}). 

In the current paper we devise a distributed \textit{deterministic} algorithm that works in the \congestmo, and for any $\epsilon>0,\rho>0$ and $\kappa=1,2,\dots,$ and for any $n$-vertex graph, it constructs a $(1+\epsilon,\beta)$-spanner with  $O_{\epsilon,\kappa,\rho}(\nfrac)$ edges in low polynomial, specifically, $O_{\epsilon,\kappa,\rho}(n^\rho)$ time, and its additive term $\beta$ is at the same ballpark as $\beta_{EN}$. Specifically, our additive term satisfies:
\begin{equation}\label{eq beta}
\beta = \bbeta.
\end{equation}

See Table \ref{table near-additive c d} for a concise comparison of our algorithm with the only previously-existing distributed deterministic \congest-model algorithm for constructing near-additive spanners  \cite{Elkin01}. See also Table \ref{table near-additive} in Appendix \ref{sec append prev res} for a concise overview of distributed algorithms, both deterministic and randomized, in both \local\ and \congest\ models, for computing near-additive spanners.

\subsection{Technical Overview and Related Work}

Our algorithm builds upon Elkin-Peleg's superclustering-and-interconnection approach (see \cite{ElkinP01,ElkinN16,ElkinN17} for an elaborate discussion of this technique). Its randomized distributed implementation by Elkin and Neiman \cite{ElkinN17} relies on random sampling of clusters. Those sampled clusters join nearby unsampled clusters to create superclusters, and this is iteratively repeated.

Our basic idea is to replace the random sampling by constructing ruling sets (see Section \ref{sec def ruling setes} for definition) using a deterministic procedure by Schneider et al. \cite{sew}. Recently, Ghaffari and Kuhn \cite{GhaffariK18} used network decompositions for derandomizing algorithms that construct multiplicative spanners, hitting and dominating sets. 
All existing deterministic constructions of network decompositions \cite{AwerbuchGLP89,PanconesiS96,BarenboimEG15} employ ruling sets. On the other hand,
 our approach avoids constructing network decomposition, but rather \textit{directly} uses ruling sets for derandomization.

\begin{table}[H]
	\centering
	\begin{adjustbox}{width=1\textwidth}
		
		\begin{tabular}{|l|l|l|l|}
			\hline
			\textbf{authors}&
			\textbf{stretch}	& 	
			\textbf{size} &	
			\textbf{running time}\\ 
			\hline

			\hline
			\cite{Elkin05}	&
			$(1+\epsilon,\beta_E)$,
			$\beta_E= \left(\frac{\kappa}{\epsilon}\right)^{O({\log \kappa})}\cdot
			(\rho^{-1})^{\rho^{-1}} $
			& 	
			$\tilde{O}(\beta_E\cdot \nfrac)$&	
			$O(n^{1+\frac{1}{2\kappa}})$\\

			\hline
			\hline
			\textbf{New}&
			$\left(1+\epsilon,\beta\right)$, 	$\beta=			 
			\bbeta $	& 	
			$ O\left(\beta\cdot\nfrac\right)$ &	
			$O\left(\beta\cdot n^\rho\cdot\rho^{-1}\right)$\\

			\hline
		\end{tabular}
	\end{adjustbox}
	\caption{Comparison of existing and new deterministic \congest-model algorithms for constructing near-additive spanners.}
		 \label{table near-additive c d}	
\end{table}

\subsection{Outline}

Section \ref{sec preliminaries} provides necessary definitions for understanding this paper. Section \ref{sec construction of spanners} contains our spanner construction along with an analysis of the running time, size and stretch of the spanner. 
Appendix \ref{section Detecting popular clusters}  contains  a variant of the Bellman-Ford algorithm that is used by the construction in Section \ref{sec construction of spanners}. Finally, Appendix \ref{sec append prev res} contains a table 
that summarizes known results concerning near-additive spanners.


\subsection{Preliminaries }\label{sec preliminaries}

\subsubsection{The Computational Model}\label{sec Computational Model}
In the distributed model \cite{Peleg00book} we have processors residing in vertices of the graph. The processors communicate with their graph neighbors in synchronous rounds. In the \congestmo, messages are limited to $O(1)$ words, i.e., $O(1)$ edge weights or ID numbers. In the \localmo, the length of messages is unbounded. 
The running time of an algorithm in the distributed model is the worst case number of communication rounds that the algorithm requires.
Throughout this paper, we assume that all vertices have unique IDs such that for all $v,\ v.ID\in [n]$, and all vertices know their ID. Moreover, we assume that all vertices know the number of vertices $n$. In fact, our results apply even if vertices know an estimate $\tilde{n}$ for $n$, where $n\leq \tilde{n}\leq poly(n)$.

\subsubsection{Ruling Sets}\label{sec def ruling setes}

Given a graph $G= (V,E)$, a set of vertices $W\subseteq V$ and parameters $\zeta,\eta \geq 0$, a set of vertices $A\subseteq W$ is said to be
a \textit{$(\zeta,\eta)$-ruling set} for $A$ if for every pair of vertices $u,v\in A$, the distance between them in $G$ is at least $\zeta$, and for every $u\in W$ there exists a representative $v\in A$ such that the distance between $u,v$ is at most $\eta$.

\section{A Deterministic Distributed Construction of Near-Additive Spanners}\label{sec construction of spanners}
\subsection{Overview}
Let $G=(V,E)$ be an unweighted, undirected graph on $n$ vertices, and let $\epsilon>0$, $\kappa=1,2,\ldots$ and ${1}/{\kappa}\leq \rho<{1}/{2}$. 
The algorithm constructs a sparse $(1+\epsilon,\beta)-spanner $ $H=(V,E_H)$ of $G$ where $\beta(\kappa,\epsilon,\rho)$ is a function of the parameters $\kappa,\epsilon,\rho$, given by \cref{eq beta},
and $|H|=O(\beta \nfrac) $, in deterministic time ${O}_{\epsilon,\kappa,\rho}(n^\rho)$ in the CONGEST model. 

The overall structure of our algorithm is reminiscent of that in \cite{ElkinN17}. 
However, unlike the algorithm of \cite{ElkinN17}, the current algorithm does not use any randomization. We first provide a high-level overview of the algorithm.

The algorithm begins by initializing $E_H$ as an empty set, and proceeds in phases. It begins by partitioning $V$ into singletons clusters $P_0= \{\{v\} \ | \ v\in V\}$. Each phase $i$, for $i=0,\ldots,\ell$, receives as input a collection of clusters $P_i$, and distance and degree threshold parameters $\delta_i,\deg_i$. The parameters $\delta_i$ and $ deg_i$ will be specified later. 
We set the maximum index of a phase $\ell$ by $\ell= \floor*{{\log \kappa\rho}}+ \lceil\frac{\kappa+1}{\kappa\rho}\rceil-1$, as in \cite{ElkinN17}.

Each of the clusters constructed by our algorithm will have a designated \textit{center} vertex.
Throughout the algorithm, we denote by $r_C$ the center of the cluster $C$ and say that $C$ is \textit{centered around} $r_C$. For a cluster $C$, define $Rad(C)= max\{d_H(r_C,v)\ | \ v\in C\}$, and for a set of clusters $P_i$, define 
$Rad(P_i)=max\{Rad(C)\ | \ C\in P_i\}$. For a collection $P_i$, we denote by $S_i$ the set of centers of $P_i$, i.e., $S_i= \{r_C\ |\ C\in P_i\}$.

Intuitively, in each phase we wish to add paths between pairs of cluster centers that are close to one another. However, if a center has many centers close to it, i.e., it is \emph{popular}, this can add too many edges to the spanner. 
Two cluster centers $r_C,r_{C'}$ are said to be \emph{close}, if $d_G(r_C,r_{C'})\leq \delta_i$. A cluster $C$ and its center $r_C$ are said to be \emph{popular} if $r_C$ has at least $deg_i$ cluster centers that are close to it. In order to avoid adding too many edges to the spanner, each phase consists of two steps, the \textit{superclustering} step and the \textit{interconnection} step.

\begin{figure}
	\begin{center}
		\includegraphics[scale=0.25]{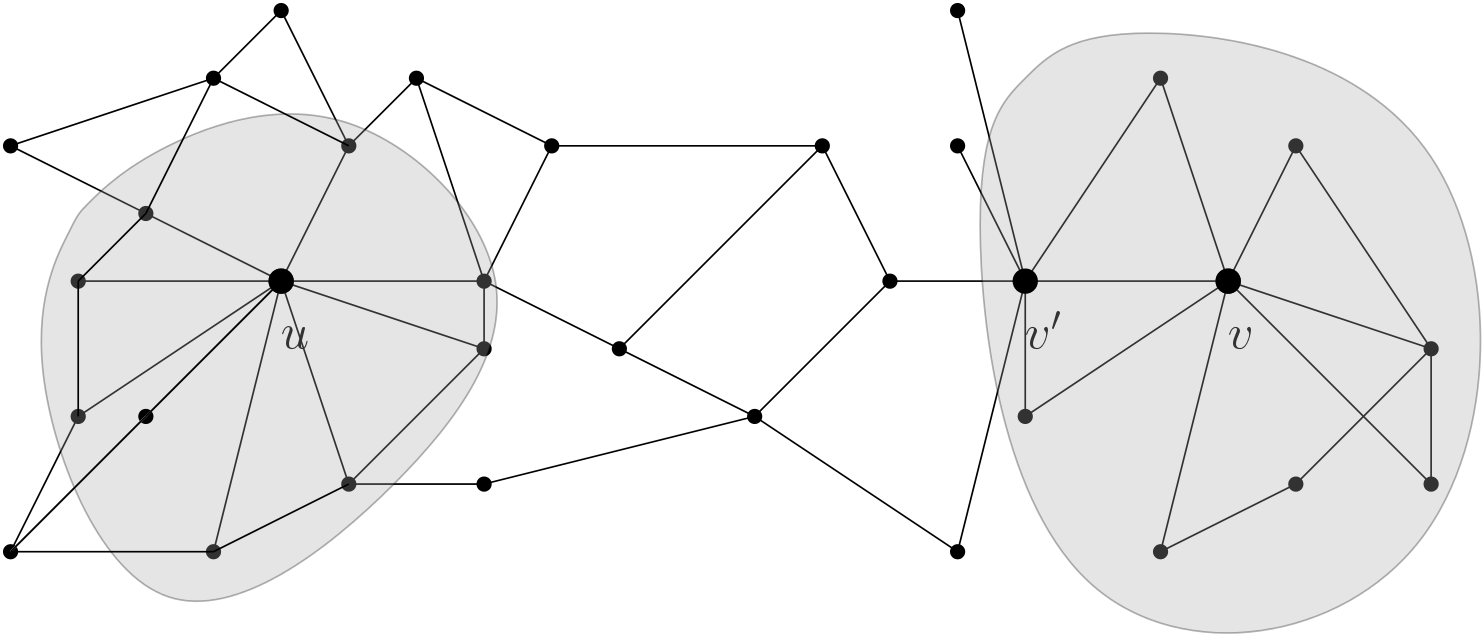}
		\caption{Superclusters centered at the chosen popular cluster centers $u,v$ are grown. (Vertices in the gray area of $u,v$ will join the superclusters of $u,v$, respectively). The popular cluster center $v'$ is covered by the supercluster centered at $v$.}
		\label{fig form clusters}
	\end{center}
	
\end{figure}

The \emph{superclustering} step of phase $i$ detects popular clusters, and builds larger clusters around them (see Figure \ref{fig form clusters}
for an illustration). For each new cluster $C$, a BFS tree of the cluster $C$ is added to the spanner $H$ (see Figure \ref{fig create bfs}).
 The collection of new clusters is the input for phase $i+1$. This allows us to defer the work on highly dense areas in the graph to later phases of the algorithm.

\begin{figure}
	
	\begin{center}
		\includegraphics[scale=0.25]{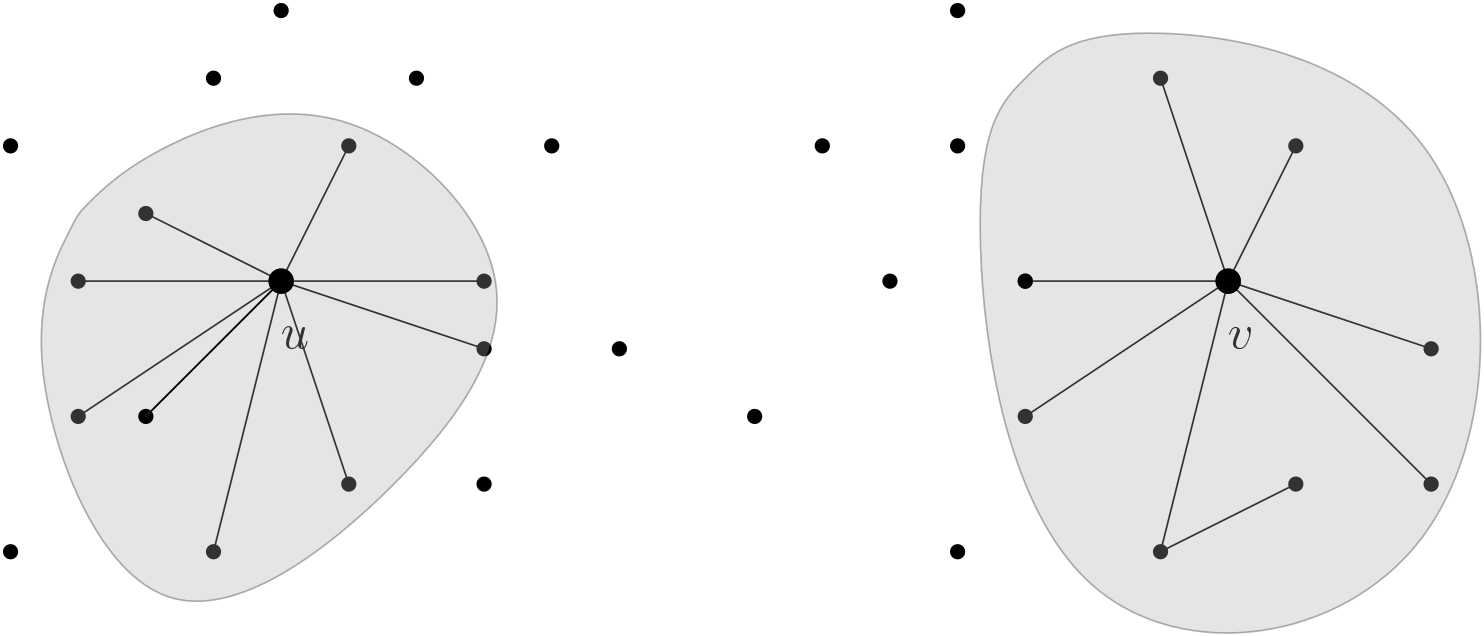}
		\caption{BFS trees of the new superclusters, rooted at the centers of the superclusters, are added to $H$.}
		\label{fig create bfs}
	\end{center}
\end{figure}

In the \textit{interconnection} step of phase $i$, clusters that have not been superclustered in this phase are connected to one another. 
For each center $r_C$ that is not superclustered in this phase, paths are added to all centers of clusters that are close to it. As the center $r_C$ is not superclustered, it is not \textit{popular}, and so we will not add too many paths to the resulting spanner $H$. 

In the last phase $\ell$, the superclustering step is skipped and we move directly to the interconnection step. 
We will show that the number of clusters in the final collection of clusters $P_\ell$ is small. Specifically, we show that even if every pair of clusters in $P_\ell$ is interconnected by a path, the number of added edges will still be relatively small. Thus the superclustering step of phase $\ell$ can be safely skipped.
This concludes the high-level overview of the algorithm.

The distance parameters are defined as follows. Define

\begin{equation}\label{def_ri+1}
R_0=0 \ \textit{ and } \ R_{i+1} = {2\epsi}/{\rho}+({5}/{\rho})R_i.
\end{equation}

We will show (see Lemma \ref{ri_bounds_pi lemma} below) that $R_i$ is an upper bound on the radius of clusters in $ P_i$, that is, $Rad( P_i)\leq R_i$ for all $0\leq i\leq \ell$. Also, the distance threshold parameter $\delta_i$ is given by: 
\begin{equation}\label{eq def deltai}
\delta_i = \epsi+2R_i.
\end{equation}

The degree threshold $deg_i$ affects the number of clusters in each phase, and thus the number of phases of the algorithm. It also affects the number of edges added to the spanner by the interconnection step. Ideally, we would like to set $deg_i= n^\frac{2^i}{\kappa}$ for all $i$.
However, the algorithm requires $\Omega(deg_i)$ time to execute phase $i$. Thus, we must keep $deg_i\leq n^\rho$ for all phases, as we aim at running time of roughly $O(n^\rho)$.

For this reason, we partition phases $0,1,\dots,\ell-1$ into two stages, the \emph{exponential growth stage} and the \emph{fixed growth stage}. In the \textit{exponential growth stage}, that consists of phases 
$ 0,\ldots,i_0=\floor*{{\log (\kappa\rho)}}$, we set $deg_i = n^\frac{2^i}{\kappa}$. In the \textit{fixed growth stage}, which consists of phases $ i_0+1,\ldots,i_1 = i_0+\lceil\frac{\kappa+1}{\kappa\rho}\rceil-2=\ell-1$, we set $deg_i = n^\rho$. Observe that for every index $i$, we have $deg_i\leq n^\rho$. The concluding phase $\ell$ is not a part of either these stages, as the number of clusters in $P_{\ell}$ is at most $n^\rho$. For notational purposes, we define $deg_\ell=n^\rho$. However, we note that as there are at most $n^\rho$ clusters in $P_\ell$, there are no popular clusters in this phase.

We will now discuss the differences between the current algorithm and the algorithm of \cite{ElkinN17}. In \cite{ElkinN17}, the superclustering step uses a randomized selection of vertices to cover the popular cluster centers. The current algorithm replaces this selection with a deterministic procedure that computes a ruling set that covers these centers. For this aim, we use the algorithm of \cite{sew,KuhnMW18} that computes ruling sets efficiently in the \congestmo.
As a result of the deterministic procedure, the distance parameter (and as a result, the radii of clusters) for each phase in the current algorithm is larger than in \cite{ElkinN17}. For this reason, the additive term that the current algorithm provides is slightly inferior to the additive term of \cite{ElkinN17}. The number of phases and the degree parameter sequence remain as in \cite{ElkinN17}. 

The deterministic detection of popular clusters' centers requires vertices to acquire information regarding their $\delta_i$-neighborhood. This information is later utilized by our algorithm in the interconnection step as well. Thus, while in the algorithm of \cite{ElkinN17} the interconnection step executes Bellman-Ford explorations, the current algorithm relies on already acquired information. Therefore, the current execution of the interconnection step is simpler than the execution of the interconnection step in \cite{ElkinN17}. 


\subsection{Superclustering}

This section provides details of the execution of the superclustering steps for all phases $i\in [ 0,i_1= \ell-1 ]$, i.e., all phases other than the (concluding) phase $\ell$ on which there is no superclustering step. 
Recall that for the exponential growth stage, we have $deg_i=n^\frac{2^i}{\kappa}$. For the fixed growth stage, we have $deg_i = n^{\rho}$.
The input to phase $i$ is a set of clusters $P_i$. The phase begins by detecting popular clusters. To do so, we employ Algorithm \ref{Alg number of near neighbors}, described in Appendix \ref{section Detecting popular clusters}, with the input $(G,P_i,\deg_i,\delta_i)$.
The following theorem, which is proven in Appendix \ref{section Detecting popular clusters}, summarizes the properties of the returned vertex set $W_i$. 

\begin{theorem}\label{theorem popular} 
	Given a graph $G=(V,E)$, a collection of clusters $P_i$ centered around cluster centers $S_i$ and parameters $\delta_i,\ \deg_i$, Algorithm \ref{Alg number of near neighbors} returns a set $W_i$ in $O(deg_i\cdot\delta_i)$ time such that: 
	\begin{enumerate}
		\item $W_i$ is the set of all centers of popular clusters from $P_i$.
		\item Every cluster center $r_C\in S_i$ that did not join $W_i$ knows the identities of all the centers $r_{C'}\in S_i$ such that $d_G(r_C,r_{C'})\leq \delta_i$. Furthermore, for each pair of such centers $r_C,r_{C'}$, there is a shortest path $\pi$ between them such that all vertices on $\pi$ know their distance from $r_{C'}$.
	\end{enumerate} 
	
\end{theorem}

Next, we wish to select a subset of vertices from $W_i$ to grow large clusters around them. On the one hand, the number of clusters in $P_{i+1}$ must be significantly smaller than the number of clusters in $P_{i}$. On the other hand, all cluster centers in $W_i$ must be superclustered. In the centralized version of this algorithm, Elkin and Peleg \cite{ElkinP01} run multiple consecutive scans to detect popular clusters and build superclusters around them. In the distributed model of computation, this requires too much time. In the distributed randomized version of this algorithm, Elkin and Neiman \cite{ElkinN17} randomly select centers to grow superclusters around. In this deterministic distributed version, we select vertices by constructing a \emph{ruling set} for the set $W_i$. We use the algorithm given in \cite{sew,KuhnMW18},
on the set $W_i$ with parameters $q=2\delta_i,c=\rho^{-1}$.
The following theorem summarizes the properties of the returned ruling set $RS_i$. 

\begin{theorem}
	\label{theorem ruling set}
	{\normalfont \cite{sew,KuhnMW18}}
	Given a graph $G=(V,E)$, a set of vertices $W_i\subseteq V$ and parameters $q\in\{1,2,\ldots \},c>1$, 
	one can compute a $(q+1,cq)$-ruling subset for $W_i$ in $O(q\cdot c\cdot n^\frac{1}{c})$ deterministic time, in the \congestmo. 
\end{theorem}

By Theorem \ref{theorem ruling set}, the returned subset $RS_i$ is a $(2\delta_i+1,({2}/{\rho})\cdot\delta_i)$\textit{-ruling set} for the set $W_i$.

The ruling set $RS_i$ is  $(2\delta_i+1)$-separated. This is done in order to guarantee that 
the sets of vertices in radius $\delta_i$ around each vertex in $RS_i$ are pairwise disjoint (for an illustration, see Figure \ref{figure o-o}).
This allows us to bound the size of the ruling set $RS_i$, and eventually the size of the collection $P_{i+1}$.

\begin{wrapfigure}{r}{0.4\textwidth}
	
	\begin{center}
		\includegraphics[scale=0.16]{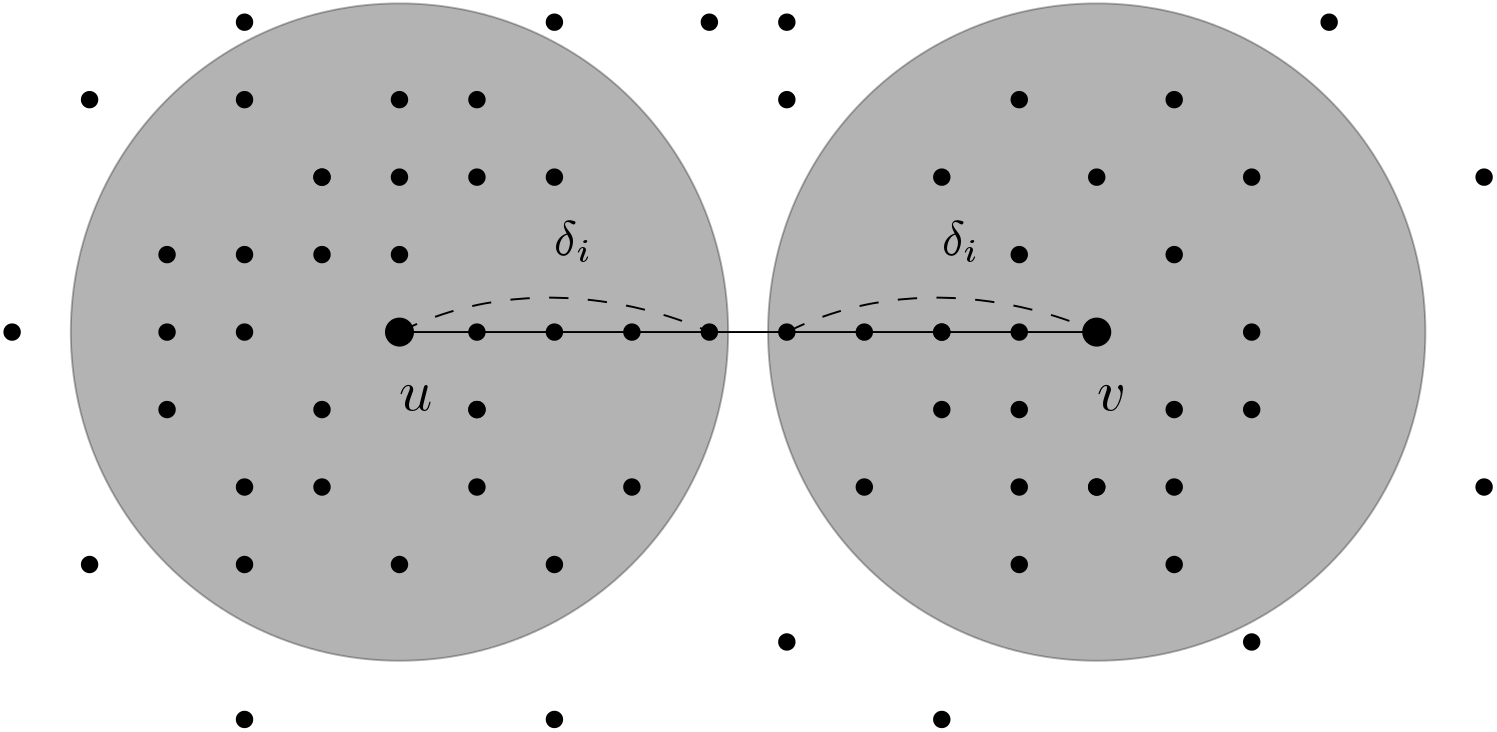}
		\caption{For every pair of vertices in $RS_i$, their $\delta_i$-neighborhoods are disjoint. In the figure, $u,v$ are two cluster centers from $RS_i$. The gray areas around them represent their respective disjoint $\delta_i$-neighborhoods.}
		\label{figure o-o}	
	\end{center}
\end{wrapfigure}

We are now ready to create large \textit{superclusters}. A BFS exploration rooted at the set $RS_i$ is executed to depth $({2}/{\rho})\cdot\delta_i$ in $G$. As a result a forest $F_i$ is constructed, rooted at vertices of $RS_i$.

For a cluster center $r_{C'}\in S_i\backslash RS_i$ that is spanned by $F_i$, let $r_C$ be the root of the forest tree of $F_i$ to which $r_{C'}$ belongs. The cluster $C'$ now becomes \textit{superclustered} in the cluster $\widehat{C}$ centered around $C$. 

The center $r_C$ of $C$ becomes the new cluster center of $\widehat{C}$, i.e., $ r_{\widehat{C}}\gets r_C$. The vertex set of the new supercluster $\widehat{C}$ is the union of the vertex set of the original cluster $C$, with the vertex sets of all clusters $C'$ which are superclustered into $\widehat{C}$. We denote by $V({C})$ the vertex set of a cluster ${C}$.

For every cluster center $r_{C'}$ that is spanned by the tree in $F_i$ rooted at $r_C$, the path in $F_i$ from $r_C$ to $r_{C'}$ is added to the spanner $H$ (see Figure \ref{figure superclustering}).
Note that the path itself is not added to the cluster $\widehat{C}$. Recall that $H$ is initialized as an empty set.

We define $\widehat{P_{i}}$ as the set of new superclusters, $ \widehat{C}$, that were built by the superclustering step of phase $i$. We set $P_{i+1} = \widehat{P_{i}}$. Note that the set $S_{i+1}$ of cluster centers of $P_{i+1}$ is given by $S_{i+1}=RS_i$.


Next we show that $R_i$ is an upper bound on $Rad(P_i)$, the radius of clusters in phase $i$.

\begin{figure}[h]
	
	\begin{center}
		\includegraphics[scale=0.3]{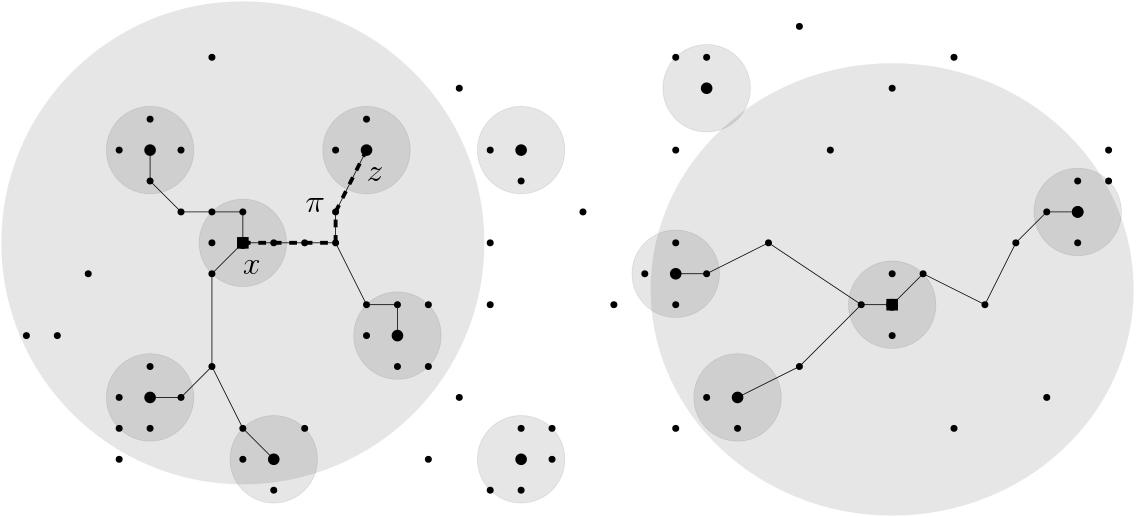}
		\caption{For every $r_{C'}$ that is spanned by the tree in $F_i$ rooted at $r_C$, the path in $F_i$ from $r_C$ to $r_{C'}$ is added to the spanner $H$. For example, in the figure, the $x-z$ path $\pi$ depicted by a thick dashed line is added to $H$.}
		\label{figure superclustering}
	\end{center}
\end{figure}

\begin{lemma}\label{ri_bounds_pi lemma}
	For all integer $0\leq i \leq \ell$, we have $Rad(P_i)\leq R_i$.
\end{lemma}

\begin{proof}
	
	We will prove the lemma by induction on the index of the phase $i$. 
	For $i=0$, observe that $Rad(P_0)= 0$, and $R_0= 0$ by definition. Thus the induction base case holds. For the inductive step, recall that by definition
	$R_{i+1}= {2\epsi}/{\rho}+\left({5}/{\rho}\right)R_i$. 

	For analyzing $Rad(\widehat{P_{i}})$, consider a vertex $u\in \widehat{C}$. If $u\in C$, by the induction hypothesis, we have $d_H(r_C,u)\leq R_i$, and since $\rho<1$ we have that $R_i\leq R_{i+1}$. Otherwise, $u\in C'$ for some $C'\in P_i$ such that $C'$ is clustered into $\widehat{C}$ in the superclustering step of phase $i$. By the induction hypothesis, we have $d_H(r_{C'},u)\leq R_i$. Since $C'$ is clustered into $\widehat{C}$, we have $d_G(r_C,r_{C'})\leq \frac{2}{\rho}\delta_i$. Moreover, a shortest $r_C-r_{C'}$ path was added to $H$, and by 
	\cref{eq def deltai}, we have
	$d_H(r_C,r_{C'})\leq \frac{2}{\rho}\delta_i = \rho^{-1}\left[2\epsi+ 4R_i\right]$. Since  $\rho< 1$, we have $
	d_H(r_C,u) \ 
	\leq\ R_{i+1}$.
	Therefore, $ Rad(P_{i+1}) = Rad(\widehat{P_i})\leq R_{i+1}$. 
\end{proof}

\begin{lemma}\label{popular are clustered}
	All \textit{popular} clusters in phase $i$ are superclustered into clusters of $\widehat{P}_i$.
\end{lemma}
%

\begin{proof}
	Let $C'$ be a popular cluster. Then, $r_{C'}$ belongs to the set $W_i$ returned by Algorithm \ref{Alg number of near neighbors}.
	There is a vertex $r_C$ in the ruling subset $RS_i$ computed for the set $W_i$, such that $d_G(r_C,r_{C'})\leq \frac{2}{\rho}\delta_i$. Hence the BFS exploration that is executed from the set $RS_i$ to depth $\frac{2}{\rho}\delta_i$ reaches $r_{C'}$. Thus $r_{C'}$ is spanned by the forest $F_i$, and it is superclustered into some supercluster in $\widehat{P}_i$. 
\end{proof}

This concludes the analysis of the superclustering step of phase $i$.


\subsection{Interconnection}

\begin{figure}[h]
	\begin{center}
		
		\includegraphics[scale=0.25]{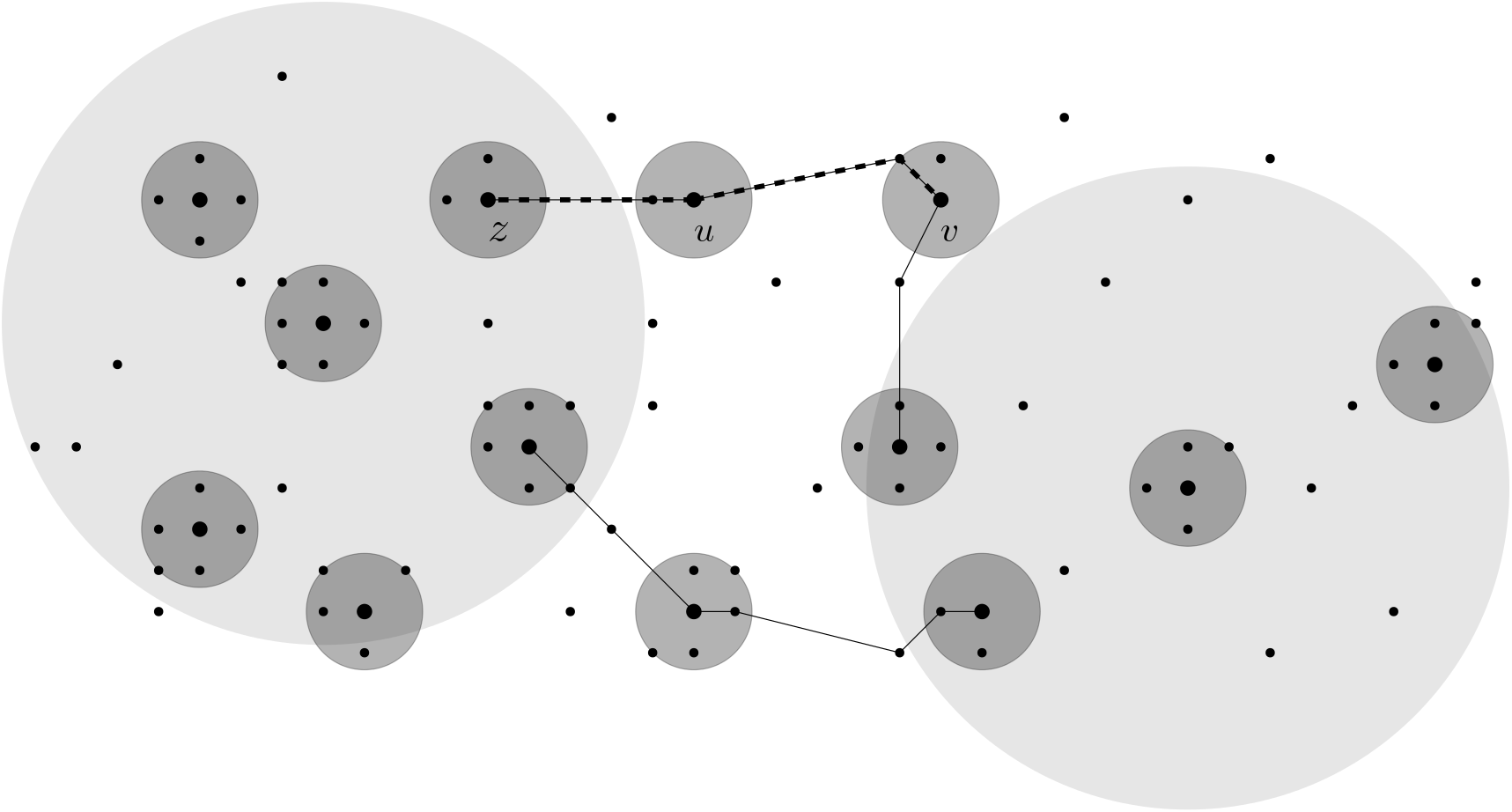}
		\caption{Each cluster center $r_C$ such that $C\in U_i$ is connected to all cluster centers $r_{C'}$, such that $C'\in P_i$ that are close to it. For example, in the figure, the paths $u-v,u-z$ depicted by a thick dashed line are added to $H$.}
		\label{figure interconnection}
		
	\end{center}
\end{figure}
Next we provide the details of the execution of the interconnection step.
Let $i\in [ 0, \ell]$. 
Denote by $U_i$ the set of clusters of $P_i$ which were not superclustered into clusters of $\widehat{P}_i$. For phase $\ell$, the superclustering step is skipped. Therefore, we set $U_\ell = P_\ell$.

In the interconnection step for $i\geq 0$, we wish to connect clusters $C\in U_i$ to all the clusters $C'\in P_i$ that are close to them, i.e., such that $d_G(r_C,r_{C'})\leq \delta_i$. Note that we will connect a cluster $C\in U_i$ to all clusters $C'\in P_i$ that are close to it. 
This is regardless of whether $C'$ has been superclustered in this phase or not (see Figure \ref{figure interconnection}).
However, by Lemma \ref{popular are clustered}, every cluster in $ U_i$ is not popular. As $C\in U_i$, it has at most $deg_i$ clusters close to it. Thus we can connect it to these other nearby clusters without adding too many edges to the spanner. 
By Theorem \ref{theorem popular}, the center $r_C$ of each such cluster $C$ knows all other centers $r_{C'}$ such that $C'\in P_i$ and $d_G(r_C,r_{C'})\leq \delta_i$, and the distances to them. Thus, each center $r_C$ of a cluster $C\in U_i$ already knows which are the clusters it needs to connect to. For each $C'\in P_i$ such that $d_G(r_C,r_{C'})\leq \delta_i$, the center $r_C$ traces back the message that informed $r_C$ regarding $r_{C'}$ (see Theorem \ref{theorem popular}), and a shortest path between $r_C$ and $r_{C'}$ is added to the spanner $H$.

The interconnection of the concluding phase $\ell$ is slightly different. As the superclustering step of this phase is skipped, we employ Algorithm \ref{Alg number of near neighbors}, described in Appendix \ref{section Detecting popular clusters}, with the input $(G,P_\ell,deg_\ell,\delta_\ell)$. The output
set $W_i$ is an empty set\footnote{
	Algorithm \ref{Alg number of near neighbors} does not explicitly \textit{return} the set $W_i$. Rather each vertex $v\in V$ knows whether it belongs to the output $W_i$ or not. In the execution of Algorithm \ref{Alg number of near neighbors} in phase $\ell$, all vertices know that they do not belong to the output $W_i$.},
 but the execution of the algorithm allows vertices to acquire the required information for the interconnection step. This completes the description of the interconnection step.

Denote by $U^{(i)}$ the union of all sets $U_0,U_1,\dots,U_i$, i.e., $U^{(i)} = \bigcup_{j=0}^iU_j$. 
The following corollary
shows that the set $U^{(\ell)}$ is a partition of $V$. 

\begin{corollary}\label{coro partition}
The set $U^{(\ell)}$ defined by $U^{(\ell)}= \bigcup_{i=0}^{\ell}U_i$ is a partition of $V$.
\end{corollary}

\newcommand{\proofpartition}{
	
	The following definitions will be used in the proof of Lemma \ref{lemma ui partition} below. 
	Denote by $VU_i$ the set of vertices $v$ such that there is a cluster $C\in U_i$ that contains $v$, i.e., $VU_i = \{ v\ | \ \exists C\in U_i \textit{ such that } v\in C \}$.
	Denote by $VP_i$ the set of vertices $v$ such that there is a cluster $C\in P_i$ that contains $v$, i.e., $VP_i = \{ v\ | \ \exists C\in P_i \textit{ such that } v\in C \}$.
	Denote $VU^{(i)}$ the union of all sets $VU_0,VU_1,\dots,VU_i$,. i.e., $VU^{(i)}= \{ v\ | \ \exists C\in U^{(i)}, \textit{ such that } v\in C \}$. For notational convenience, we will also define $P_{\ell+1}$, $VP_{\ell+1}$ and $VU_{\ell+1}$ to be empty sets.
	
	\begin{lemma}\label{lemma ui partition}
		
		For every index $i\in [0,\ell]$, the set $U^{(i)}$ is a partition of $VU^{(i)}$. 
	\end{lemma}

\begin{proof}

	The proof is by induction on the index of the phase $i$. 
	
	For $i=0$, the set $U_0$ is the set of all clusters that were not superclustered in phase $0$. As all these clusters are singletons, we have that $U_0$ is a partition of $VU^{(0)}$ to singletons. 
	Note that every vertex $v\in V$ that does not belong to $VU^{(0)}$, belongs to some cluster $C\in P_1$.

	Assume that $U^{(i-1)}$ is a partition of $VU^{(i-1)}$, for some index $i$ in the range $ [1,\ell]$. We will prove that $U^{(i)}$ is a partition of $VU^{(i)}$.

	For the induction step, first observe that for every vertex $v$ that belongs to a cluster $\widehat{C}\in P_{i}$, there is a cluster $C\in P_{i-1}$ such that $v$ belongs to $C$. It follows that 
\begin{equation}\label{eq supersets}
	V=VP_0\supseteq VP_1\supseteq\dots \supseteq VP_\ell\supseteq VP_{\ell+1}=\emptyset.
	\end{equation}

	Moreover, each vertex can belong to at most one cluster in $P_{i-1}$. For each cluster $C\in P_{i-1}$, either $C$ has been superclustered into a cluster of $P_{i}$, or it has joined $U_{i-1}$ (note that this condition also applies for $i=\ell+1$, as $P_{\ell+1}= \emptyset$, and $U_\ell= P_\ell$). These two cases are mutually exclusive. Thus $ VU_{i-1} \cup VP_{i}= VP_{i-1} $, and $ VU_{i-1} \cap VP_{i}= \emptyset$.
	By \cref{eq supersets}, it follows that 
\begin{equation}\label{eq empty cap}
	VU^{({i-1})}\cap VP_{i}= \emptyset.
	\end{equation}

	Let $v$ be a vertex in $VU^{(i)}$. We will consider now two complementary cases. 
	
	\textbf{Case 1:} $v$ does not belong to $VP_i$. Then, $v$ does not belong to $VU_i$. Thus it belongs to $VU^{(i-1)}$, and by the induction hypothesis there is exactly one cluster $C\in U^{(i-1)}$ such that $v\in C$. Moreover, since $v$ does not belong to $VP_i$, there is no cluster $C'\in P_i$ such that $v\in C$. Therefore $v\notin VU_i$. It follows that $v$ belongs to exactly one cluster $C\in U^{(i)}$. 
	
	\textbf{Case 2:} $v$ belongs to $VP_i$. Then by \cref{eq empty cap}, $v$ does not belong to $VU^{(i-1)}$. By the induction hypothesis, since $U^{(i-1)}$ is a partition of $VU^{(i-1)}$, there is no cluster $C\in U^{(i-1)}$ such that $v$ belongs to $C$. Since $v$ does belong to $VU^{(i)}$, it follows that there is exactly one cluster $C'\in U_i$ such that $v\in C'$. Therefore, $v$ belongs to exactly one cluster $C\in U^{(i)}$. 
	
	Hence $U^{(i)}$ us a partition of $VU^{(i)}$.

\end{proof}

Consider the set $VU^{(\ell)}$. Note that $VP_0= V$, and that for every index $i\in[0,\ell]$, each vertex $v$ that belongs to $VP_i$ either belongs to $VU_i$ or to $VP_{i+1}$. Since $VP_{\ell+1}= \emptyset$ it follows that $VU^{(\ell)} = V$. As a corollary, we conclude \\

{\bf Corollary \ref{coro partition}}
	The set $U^{(\ell)}$ defined by $U^{(\ell)}= \bigcup_{i=0}^{\ell}U_i$ is a partition of $V$.

}

\proofpartition

\subsection{Analysis of the Construction}

In this section, we analyze the running time, the size and the stretch of our construction. We begin by analyzing the radii of clusters' collections $P_i$ for $i\in [\ell] $. By Lemma \ref{ri_bounds_pi lemma}, these radii are upper bounded by $R_i$'s. In the next lemma, we derive an explicit expression for these upper bounds. This upper bound is later used to bound $\delta_i$, which is used in Sections \ref{sec running time}, \ref{sec size} and \ref{sec stretch} to analyze the running time, size and stretch of the spanner, respectively. Finally, in Section \ref{sec rescale}, we rescale $\epsilon$ to derive our ultimate result.

\begin{lemma} \label{lemma ri bound not explicit}
	For $i\in [\ell] $, we have
$R_i = \sum_{j=0}^{i-1}
	{2}/{\rho}\cdot\epsilon^{-j}\cdot \left({5}/{\rho}\right)^{i-1-j}.
$
\end{lemma}

\begin{proof}
	The proof is by induction on the index $i$. 
	
	Recall that	$R_0 = 0$ and $R_i$ is given by (see \cref{def_ri+1})
	$R_{i+1} = \frac{2}{\rho}\epsi+\left({5}/{\rho}\right)R_i$.
	
	For $i=1$, it is easy to verify that
	$\sum_{j=0}^{i-1}
	{2}/{\rho}\cdot\epsilon^{-j}\cdot \left({5}/{\rho}\right)^{i-1-j} =2/\rho$ and also $ \frac{2}{\rho}\eps{0}+\left({5}/{\rho}\right)R_{0} = 2/\rho$. So the base case holds. For the induction step, by the induction hypothesis we have: 
	
\begin{equation*}
	\begin{array}{clllll}
	R_{i+1} &=& ({2}/{\rho})\cdot\epsi+\left({5}/{\rho}\right) \sum_{j=0}^{i-1}
	\frac{2}{\rho}\left(\frac{1}{\epsilon}\right)^j\left({5}/{\rho}\right)^{i-1-j}
	&=& \sum_{j=0}^{i}
	\frac{2}{\rho}\left(\frac{1}{\epsilon}\right)^j\left({5}/{\rho}\right)^{i-j}.
	\end{array}
	\end{equation*}
\end{proof}

Assume that $\rho\geq 10\epsilon$. (This assumption will not affect our ultimate result. See Section \ref{sec rescale}.) 
Observe that Lemma \ref{lemma ri bound not explicit} implies that
$R_i \leq \frac{2}{\rho-5\epsilon}\cdot\eps{(i-1)}$.

In particular, we have
$\frac{2}{\rho-5\epsilon}\cdot\eps{(i-1)}\leq\frac{4}{\rho}\cdot\eps{(i-1)}.$
It follows that: 

\begin{equation}\label{eq bound ri}
\begin{array}{cllllllllllll}
R_i &\leq& \frac{4}{\rho}\cdot\eps{(i-1)}.
\end{array}
\end{equation}

Recall that $P_i = \widehat{P}_{i-1}$. Hence $Rad(P_i) = Rad(\widehat{P}_{i-1})$. By Lemma \ref{ri_bounds_pi lemma}, we have for all $i\in[0,\ell]$, (assuming $\rho\geq 10\epsilon$),
\begin{equation}\label{eq 4 rho}
\begin{array}{cllllllll}
Rad(P_i) & = & Rad(\widehat{P}_{i-1}) & \leq & R_{i} &\leq & \frac{4}{\rho}\cdot\eps{(i-1)}.
\end{array}
\end{equation} 

Recall that $\delta_i=\epsi+2R_i$, and that $R_0=0$, and $\delta_0= 1$. For $i\in [\ell]$, by \cref{eq bound ri}, we obtain: $\delta_i \leq \epsi+ \frac{8}{\rho}\cdot\eps{(i-1)}$. Assume that $\rho\geq 10\epsilon$. Then $\frac{8}{\rho}\eps{(i-1)}\leq \epsi$. It follows that for all $i\in [0,\ell]$:

\begin{equation}\label{eq bound deltai}
\delta_i \leq O\left( (1/\epsilon)^i \right).
\end{equation}

\subsubsection{Analysis of the Running Time}\label{sec running time}

In this section, we analyze the running time of the entire algorithm. We begin with the following lemma which analyses the running time of a single phase of the algorithm.

\begin{lemma}\label{lemma rt per phase}
	For all $i\in [0,\ell]$, the running time of phase $i$ is $O\left(\rho^{-1}\cdot \delta_i\cdot n^\rho\right)$.
	
\end{lemma}

\newcommand{\proofrt}{
\begin{proof}
	
	The superclustering step consists of running Algorithm \ref{Alg number of near neighbors}, which requires $O(deg_i\cdot \delta_i)$ time 
	(see Theorem \ref{theorem popular}).
	Constructing a $(2\delta_i+1,\frac{2}{\rho}\delta_i)$-ruling set requires $O(\rho^{-1}\cdot\delta_i\cdot n^{\rho}) $ time 
	(see Theorem \ref{theorem ruling set}). The BFS exploration to depth $\frac{2}{\rho}\delta_i$ that builds superclusters requires $ O(\rho^{-1}\delta_i)$ time. (Note that no congestion occurs on this step.) In total, the superclustering step of phase $i$ requires $O(deg_i\cdot\delta_i+\rho^{-1} \delta_i\cdot n^\rho+\rho^{-1}\delta_i)$ time. 
	Since for all $i$, $deg_i\leq n^\rho$, it follows that the superclustering step of phase $i$ requires $O(\rho^{-1}\cdot \delta_i\cdot n^\rho) $ time.
	
	As for the interconnection step, note that each cluster that needs to add a path to the spanner $H$ in the interconnection step of phase $i$ knows all the clusters it needs to add a path to. Moreover, by Theorem \ref{theorem popular}, each cluster center $r_C$, $C\in U_i$, can trace back the shortest path that a message regarding a nearby cluster center $r_{C'}$ took to reach $r_C$, and add the edges along this route to $H$. By Theorem \ref{theorem popular}, this requires $O(deg_i\cdot \delta_i)$ time. Thus the running time of the interconnection step of a given phase is dominated by the running time of the superclustering step of this phase. 
	An exception to that is the concluding phase, in which there is no superclustering step. In this phase, we execute Algorithm \ref{Alg number of near neighbors} with parameters $n^\rho,\delta_\ell$. By Theorem \ref{theorem popular} this requires $O(n^\rho\cdot\delta_\ell)$ time. We then trace back the shortest paths in $O(n^\rho\cdot\delta_\ell)$ time. Thus the running time of the interconnection step of each phase $i\geq 0$ is dominated by $O(\rho^{-1}\cdot \delta_i\cdot n^\rho) $, and so the running time of a single phase of the algorithm is $O(\rho^{-1}\cdot \delta_i\cdot n^\rho) $.

\end{proof}
}
\proofrt

By Lemma \ref{lemma rt per phase} and \cref{eq bound deltai}, the running time of the entire algorithm is bounded by: 
\begin{equation}\label{eq bound running time}
\begin{array}{ccccccc}
\sum_{i=0}^\ell O\left(\rho^{-1}\cdot \delta_i\cdot n^\rho\right)
&\leq&
n^\rho\cdot\rho^{-1}\sum_{i=0}^\ell O\left( \epsi\right)
&\leq&
O\left(n^\rho\cdot \rho^{-1}\cdot \eps{\ell}\right).
\end{array}
\end{equation}
%
%
%
%

\begin{corollary}\label{coro rt}
	The running time of the algorithm is bounded by $
	O\left({n^\rho}{\rho^{-1}}\eps{\ell}\right)$.
\end{corollary}


\subsubsection{Analysis of the Number of Edges}\label{sec size}
In this section, we analyze the number of edges added to the spanner $H$. First, observe that in the superclustering step of phase $i$, the edges that are added to $H$ are a subset of the BFS forest $F_i$. Thus in each phase $O(n)$ edges are added to $H$ by the superclustering step. 

We will now analyze the number of edges added by the interconnection step. In the interconnection step of phase $i$, a path is added to the spanner from each cluster $C$ in $U_i$ to clusters in $P_i$ that are close to it. As $C$ belongs to $U_i$, it is not popular. Thus it has at most $deg_i$ other clusters in $P_i$ that are close to it. To bound the size of $U_i$, we bound the size of $P_i$, as $U_i\subseteq P_i$. The following lemma provides an upper bound on the size of $P_i$ in the exponential growth stage. 

The next two lemmas
provide an upper bound on the size of the cluster collection $P_i$ in the exponential growth and the fixed growth stages.

\begin{lemma}\label{inter_rt_lm2}
	For $i\in [0, i_0+1= \floor*{\log(\kappa\rho)}+1]$, we have 
	$| {P}_i| \leq n^{1-\frac{2^i-1}{\kappa}}.$
\end{lemma}

\newcommand{\proofboundismall}{
	\begin{proof}
	We will prove the lemma by induction on the index of the phase $i$.

	For $i=0$, the right-hand side is $n^{1-\frac{2^0-1}{\kappa}} = n$. Thus the claim is trivial. 
	
	Recall that for each index $i$, the set $RS_i$ is the ruling set computed in phase $i$, and that $S_i$ is the set of centers of clusters in $P_i$. Since for each phase $i\geq 1$, the clusters of the collection $P_i$ are centered around vertices of $RS_{i-1}$, for $i\geq 1$, we have $S_{i}= RS_{i-1}$. The set $RS_i$ is a $(2\delta_i+1,\frac{2}{\rho}\delta_i)$-ruling set for $W_i$. By Theorem \ref{theorem popular}, all vertices in $W_i$ are popular cluster centers. Thus, for every $r_C\in W_i$, it holds that
	$ |\Gamma^{(\delta_i)}(r_C)\cap S_i| \geq deg_i$. 
	
	By Theorem \ref{theorem ruling set}, the set $RS_i$ is $(2\delta_i+1)$-separated, i.e., for every pair of distinct cluster centers $ r_C,r_{C'}\in RS_i$ we have $ d_G(r_C,r_{C'})\geq 2\delta_i+1$.
	Thus, for every pair of distinct centers $r_C,r_{C'}\in RS_i$, their $\delta_i$-neighborhoods are disjoint, i.e., $
	\Gamma^{(\delta_i)}(r_C)\cap \Gamma^{(\delta_i)}(r_{C'}) = \emptyset$.
	
	Together with the induction hypothesis, and since for i, $0\leq i \leq i_0$, $deg_i = n^{\frac{2^i}{\kappa}}$, this implies that $
	| \widehat{P}_i|
	\leq \frac{| {P}_{i}|}{deg_{i}} 
	\leq  {n^{1-\frac{2^{i}-1}{\kappa}}}/{n^\frac{2^{i}}{\kappa}}
	= n^{1-\frac{2^{i+1}-1}{\kappa} }$.
	
	As $P_{i+1}= \widehat{P}_i$, we conclude that $| {P}_{i+1}| \leq n^{1-\frac{2^{i+1}-1}{\kappa}}$.
\end{proof}
}

\proofboundismall

Observe that by Lemma \ref{inter_rt_lm2}, for the phase ${i_0+1}$, we have that:

\begin{equation}\label{Bound P i 0+1}
| {P}_{i_0+1}| \quad\leq\quad n^{1-\frac{2^{i_0+1}-1}{\kappa}}.
\end{equation}

\begin{lemma}\label{inter_rt_lm5}
	For $i_0+1\leq i \leq \ell$, it holds that
	$| {P}_i| \quad \leq\quad n^{1+\frac{1}{\kappa}-(i-i_0)\rho}. $
\end{lemma}

\newcommand{\proofboundibig}{

\begin{proof}
	The proof is by induction on the index of the phase $i$. 
	
	For the base case, by \cref{Bound P i 0+1},
\begin{equation*}
	| {P}_{i_0+1}| \quad\leq \quad n^{1-\frac{2^{i_0+1}-1}{\kappa}}\quad =\quad 
	n^{1-\frac{2^{\floor*{{\log \kappa\rho}}+1}-1}{\kappa}} \quad \leq \quad 
	n^{1+\frac{1}{\kappa}-\rho}\quad =\quad 
	n^{1+\frac{1}{\kappa}-(i_0+1-i_0)\rho}.
	\end{equation*}

	For $i>i_0+1$, we have $S_{i}= RS_{i-1}$. The set $RS_i$ is a $(2\delta_i+1,\frac{2}{\rho}\delta_i)$-ruling set for $W_i$. By Theorem \ref{theorem popular}, all vertices in $W_i$ are popular, i.e., 
	for every cluster center $r_C\in W_i$, it holds that $ |\Gamma^{(\delta_i)}(r_C)\cap S_i| \geq n^\rho.$
	
	By Theorem \ref{theorem ruling set}, the set $RS_i$ is $(\delta_i+1)$-separated, i.e., for every pair of distinct cluster centers $ r_C,r_{C'}\in RS_i$, we have $ d_G(r_C,r_{C'})\geq 2\delta_i+1.$
	
	Thus, for every pair of centers $r_C,r_{C'}\in RS_i$,
	$
	\Gamma^{(\delta_i)}(r_C)\cap \Gamma^{(\delta_i)}(r_{C'}) = \emptyset$.
	Together with the induction hypothesis, this implies that

\begin{equation*}
	| \widehat{P}_i|
	\quad \leq\quad {| {P}_{i}|}/{n^\rho} 
	\quad \leq\quad { n^{1+\frac{1}{\kappa}-(i-i_0)\rho-\rho} }
	\quad = \quad n^{1+\frac{1}{\kappa}-(i+1-i_0)\rho}.
	\end{equation*}
	
	Since $P_{i+1}=\widehat{P}_i$, it follows that $| {P}_{i+1}|\leq n^{1+\frac{1}{\kappa}-(i+1-i_0)\rho}$.
\end{proof}
}
\proofboundibig

The next lemma
provides an upper bound on the number of edges added to the spanner by each phase $i\in [0,\ell]$.

\begin{lemma}\label{lemma number of edges}

For all $i\in [0,\ell]$, in each phase $i$, $O(\nfrac\cdot \delta_i)$ edges are added to the spanner $H$.
\end{lemma}

\newcommand{\proofsize}{
\begin{proof}
	
	In the interconnection step of phase $0$, we add all edges adjacent to \textit{unpopular} clusters. Note that all vertices $v\in U_0$, have at most $n^{1/\kappa}$ neighbors, i.e., $|\Gamma(v)|<n^{1/\kappa}$. Hence the number of edges added by the interconnection step of phase $0$ is $O(n^{1+\frac{1}{\kappa}})$. 
	
	Let $i\in [0, i_1]$ be an index (recall that $i_1 = i_0+\lceil\frac{\kappa+1}{\kappa\rho}\rceil-2= \ell-1$). Each cluster center $r_C$, such that $C\in U_i$, initiates a BFS exploration to depth $\delta_i$. For each cluster center $r_{C'}$, $C'\in P_i$, that is discovered by this exploration, a shortest $r_C,r_{C'} $ path is added to the spanner $H$. Since $C\in U_i$, the cluster center $r_C$ has at most $deg_i$ other cluster centers $r_{C'}\in S_i$ within distance $\delta_i$ from it. So, the number of paths that are added to the spanner in the interconnection step of phase $i$ is at most: $|U_i|\cdot deg_i\leq |P_i|\cdot deg_i$.
	
	For $i\in [0,i_0= \floor*{\log(\kappa\rho)}]$, by Lemma \ref{inter_rt_lm2}, we have $
	|P_i|\cdot deg_i \leq n^{1-\frac{2^i-1}{\kappa}}\cdot n^\frac{2^i}{\kappa} = n^{1+\frac{1}{\kappa}}$.

	For $i\in [i_0+1,i_1]$, by Lemma \ref{inter_rt_lm5}, we have $
	|P_i|\cdot deg_i \leq n^{1+\frac{1}{\kappa}-(i-i_0)\rho}\cdot n^\rho
	\leq n^{1+\frac{1}{\kappa}}$.
	
	So in each phase $i\in [0,i_1]$, the number of edges added to the spanner $H$ by the interconnection step is $O(n^{1+\frac{1}{\kappa}}\cdot\delta_i)$, and the total 
	number of edges added to $H$ by the superclustering and the interconnection steps of phase $i$ is: 

\begin{equation*}
	O(n+\nfrac\cdot\delta_i).
	\end{equation*}

	We will now discuss the number of edges added to $H$ by phase $\ell$, i.e., by the concluding phase. In this phase, we skip the superclustering step and execute the interconnection step. We set $U_\ell= P_\ell$. Note that by Lemma \ref{inter_rt_lm5}, for the last ($\ell=i_0+\lceil\frac{\kappa+1}{\kappa\rho}\rceil-1$) phase, the size of the input collection $P_\ell$ is bounded by $
	| {P}_{\ell}|\leq n^{1+\frac{1}{\kappa}-(\ell-i_0)\rho} =
	n^{1+\frac{1}{\kappa}-(i_0+\lceil\frac{\kappa+1}{\kappa\rho}\rceil-1-i_0)\rho}\leq
	n^{1+\frac{1}{\kappa}-(\frac{\kappa+1}{\kappa}-\rho)}\leq n^\rho$.
	
	Therefore, even if every pair of clusters in $P_\ell$ is connected by the interconnection step of phase $\ell$, at most $O(n^{2\rho}\cdot\delta_\ell)$ edges are added to $H$ by the concluding phase. As we set $\rho\leq \frac{1}{2}$, the concluding phase adds at most $O(n\cdot\delta_\ell)$ edges to the spanner $H$. 
	
	Hence, the total number of edges added to the spanner $H$ by a phase $i\in [0,\ell]$ is $
	O(\nfrac\cdot \delta_i)$.
	
\end{proof}

}
\proofsize

By Lemma \ref{lemma number of edges} and \cref{eq bound deltai}, the number of edges added to the spanner $H$ by all phases of the algorithm is bounded by: 
\begin{equation}\label{eq bound size}
\begin{array}{ccccccc}
\sum_{i=0}^\ell 
O\left(\nfrac\cdot \delta_i\right)
&\leq&
\sum_{i=0}^\ell 
O\left(\nfrac\cdot\epsi\right)
&\leq&
O\left(\nfrac\cdot\eps{\ell}\right)
\end{array}
\end{equation}
%

\begin{corollary}\label{coro size}
	The size of the spanner $H$ is bounded by
$
	|E_H|  =  O\left(\nfrac\cdot\eps{\ell}\right)$.
\end{corollary}


\subsubsection{Analysis of the Stretch}\label{sec stretch}

In this section, we analyze the stretch of the spanner $H$. The following two lemmas
  provide the necessary tools for Lemma \ref{lemma stretch}, in which we analyze the stretch of the spanner $H$.

\begin{lemma}\label{str1}
	For all $i\in [0,\ell]$, for every pair of clusters $C\in {U}_i$, $C'\in {P}_i$ at distance at most $\left(\frac{1}{\epsilon}\right)^i$ from one another, a shortest path between the clusters centers $r_C,r_{C'}$ was added to the spanner $H$. 
\end{lemma}

\newcommand{\proofstretcho}{
\begin{proof}
	Let $i\in [0,\ell]$, and let $C\in {U}_i$, $C'\in {P}_i$ a pair of clusters at distance at most $\left(\frac{1}{\epsilon}\right)^i$ from one another. 
	Since $R_i$ is an upper bound on $Rad(P_i)$ (see Lemma \ref{ri_bounds_pi lemma}), we have that $d_G(r_C,r_{C'})\leq 2R_i+\epsi= \delta_i$. 
	Recall that by Lemma \ref{popular are clustered}, all popular clusters in phase $i$ are superclustered into clusters of $P_{i+1}$. Thus they do not belong to $U_i$. It follows that centers of clusters of $U_i$ do not belong to $W_i$. Thus, since $C\in {U}_i$, by Theorem \ref{theorem popular}, in the interconnection step of phase $i$ the center $r_C$ of $C$ knows the distance to $r_{C'}$. So $r_C$ adds a shortest $r_C-r_{C'}$ path to the spanner $H$. 	
\end{proof}
}

\proofstretcho

We now provide an upper bound on the distance in $H$ between neighboring clusters in $G$.

\begin{lemma}\label{3rj_1_ri}
	Consider a pair of indices $0\leq j<i\leq \ell$, and a pair of neighboring clusters $C\in U_j$ and $C'\in U_i$. Let $w$ be some vertex in $C$, and let $r_{C'}$ be the center of the cluster $C'$. Let $\rho<1$. Then, there is a path in $H$ between $w$ and $r_{C'}$ of length at most $3R_j+1+R_i$.
\end{lemma}

\newcommand{\proofthreeradi}{
\begin{proof}
	
	Let $(z',z)\in E$ be an edge such that $z\in C$ and $z'\in C'$ (see Figure \ref{fig 3rad} for an illustration). We will first consider phase $j$. Let $C''\in P_j$ be the cluster such that $z'\in C''$, that is, $C''$ is the cluster of $z'$ in the $j$th phase. Observe that $C\in U_j$, and $d_G(C,C'') = 1$. Then by Lemma \ref{ri_bounds_pi lemma}, we have $d_G(r_C,r_{C''})\leq 2R_j+1$. Note that $ 2R_j+1\leq \eps{j}+2R_j= \delta_j$, for all $j\in [0,\ell]$. Thus, a shortest $r_C-r_{C''}$ path was added to $H$ in the interconnection step of phase $j$. By Lemma \ref{ri_bounds_pi lemma}, there is a path from $w$ to $r_C$ in $H$ of length at most $R_j$, and there is a path from $r_{C''}$ to $r_{C'}$ of length at most $R_i$. Note also that $\eps{j}\leq R_j$. Thus, 	
$
	d_H(w,r_{C'})
	 \leq  d_H(w,r_C)+d_H(r_C,r_{C''})+d_H(r_{C''},r_{C'})
	\leq 3R_j+1+ R_i$.
	
	\begin{figure}[H]
		\begin{center}
			\includegraphics[scale=0.3]{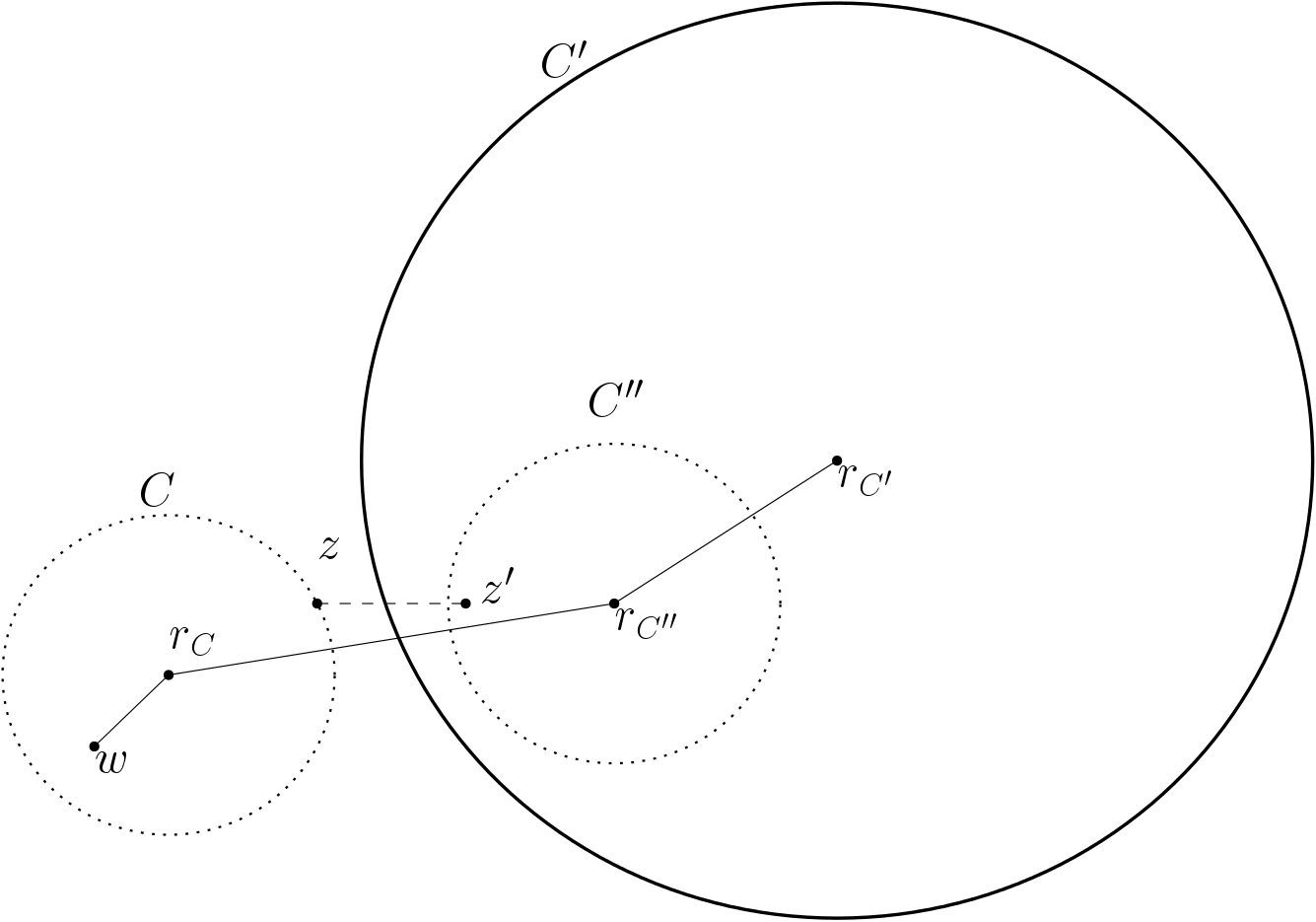}
			\caption{The path in the spanner $H$ from $w$ to $R_{C'}$. In the figure, the dotted ovals represent clusters in $U_j$, and the large oval represents a cluster in $U_i$. The dashed line represents an edge in $G$, and the solid lines represent the path in the spanner $H$.} 		
			\label{fig 3rad}
		\end{center}
	\end{figure}
	
\end{proof}
}
\proofthreeradi

Since $ R_{i+1} = \frac{2}{\rho}\epsi +\left({5}/{\rho}\right)R_i$, and since $\rho<1$ and $j<i$, we have $3R_j \leq R_i$. It follows that: 
\begin{equation}\label{3rjri_leq_2ri}
\begin{array}{cllll}
d_H(w,r_{C'})&\leq & 3R_j+1+R_i
& \leq & 2R_i +1.
\end{array}
\end{equation}

We are now ready to analyze the stretch of our spanner. Recall that for every index $i$, the set $U^{(i)}$ is the union of all sets $U_0,U_1,\dots,U_i$.

\begin{lemma}\label{lemma stretch}
	
	Assume $\epsilon\leq \frac{1}{10}$ and $\rho\geq 10\epsilon$. Consider a pair of vertices $u,v\in V$. Fix a shortest path $\pi(u,v)$ between them in $G$, and suppose that for some $i\leq \ell$ all the vertices of $\pi(u,v)$ are all clustered in the set $U^{(i)}$. Then,
\begin{equation}
	d_H(u,v) = \left(1+(30\cdot \epsilon \cdot i){\rho^{-1}}\right)d_G(u,v) + 
	6\sum_{j=0}^{i} R_j\cdot 2^{i-j}
	\end{equation}

\end{lemma}

\begin{proof}
	
	The proof is by induction on $i$. For the base case, $i=0$, all the vertices on 
	$\pi(u,v)$ are $U_0$-clustered, and so all the edges of the path are in the spanner, and $d_G(u,v)=d_H(u,v)$.

	For the induction step, we first consider a pair of vertices $x,y$ such that there is a path $\pi(x,y)$ in $G$ of length at most $ \epsi$, and such that all the vertices on the path are clustered in the set $U^{(i)}$. 
	
		\begin{figure}
	\begin{center}

				\includegraphics[scale=0.4]{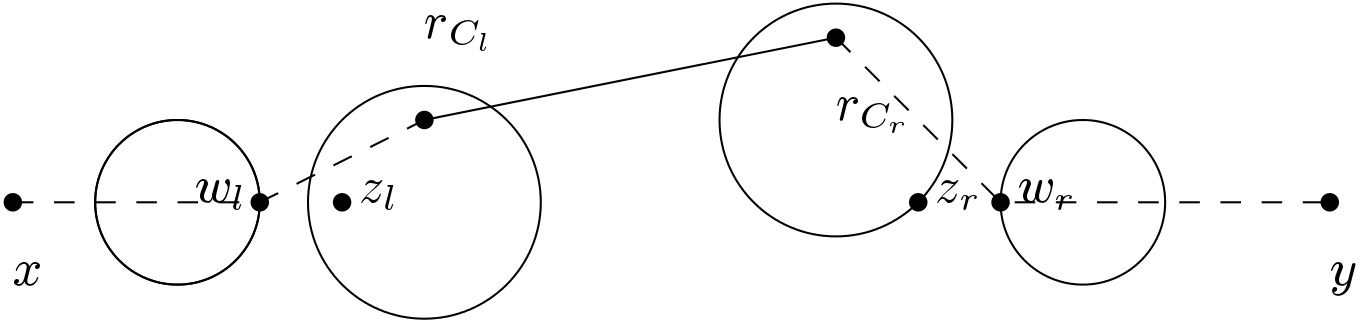}
				\caption{The path in the spanner $H$ from $x$ to $y$. In the figure, the dashed lines represent the path in the spanner $H$ that is not necessarily a shortest path in $G$. The solid line represents the path in $H$, that is also a shortest path in $G$.}
				\label{x y dist fig}	
		\end{center}
			
		\end{figure}
	
	To simplify presentation, we will imagine that the vertices of $\pi(x,y)$ appear from left to right (see Figure \ref{x y dist fig} 
	for an illustration), where $x$ (respectively, $y$) is the leftmost (respectively, rightmost) vertex of the path. Let $z_l,z_r$ be the leftmost and rightmost $U_i$-clustered vertices on $\pi(x,y)$, respectively, and let $C_l,C_r$ ($r_{C_l},r_{C_r}$) be their respective clusters (cluster centers).
	
	We note that it is possible that $z_l,z_r$ belong to the same cluster, i.e, $C_l=C_r$. However, this case is simpler, as in the analysis the path between the cluster centers $r_{C_l},r_{C_r}$ will simply be an empty path. We also note that it is possible that there is no $U_i$-clustered vertex on the path $\pi(x,y)$. In this case, we can take $i'$ to be the maximal index such that there is a vertex on $\pi(x,y)$ that is $U_{i'}$-clustered. Replacing the index $i$ with an index $i'$ such that $i'<i$ will only decrease the bound on the distance between $x,y$ in $H$. Therefore, without loss of generality, we assume that there exists a vertex on $\pi(x,y)$ that is $U_i$ clustered. 
	Let $w_l$ (respectively, $w_r$) be the neighbor of $z_l$ (respectively, $z_r$) on the sub-path $\pi(x,z_l)$ (respectively, $\pi(z_r,y)$). Observe that $w_l,w_r$ are $U^{(i-1)}$ clustered. By \cref{3rjri_leq_2ri}, $
	d_H(w_l,r_{C_l})\leq 2R_i+1 \ \text{ and }\ d_H(r_{C_r},w_r)\leq 2R_i+1$. 
	
	Note that as $|\pi(x,y)|\leq \epsi$, we have:
\begin{equation*}
	\begin{array}{lllllllllll}
	d_G(r_{C_l},r_{C_r})& \leq d_G(r_{C_l},z_l) + d_G(z_{l},z_r) +d_G(z_r,r_{C_r})
	&\leq 
	d_G(z_{l},z_r)+2R_i
	&\leq  \epsi +2R_i
	& = \delta_i.
	\end{array}
	\end{equation*}

	From this and from the fact that $C_l,C_r\in U_i$, it follows that 
	a shortest path between $r_{C_l},r_{C_r}$ was added to the spanner $H$ in the interconnection step of phase $i$. Thus $d_G(r_{C_l},r_{C_r}) =d_H(r_{C_l},r_{C_r})$, and $
	d_H(r_{C_l},r_{C_r})\leq d_G(z_l,z_r)+2R_i$.

	As all the vertices in the sub-paths $\pi(x,w_l),\pi(w_r,y)$ are $U^{(i-1)}$ clustered, the induction hypothesis applies to them. Hence
	
\begin{equation*}
	\begin{array}{lllllllllll}
	d_H(x,w_l)&\leq& 
	\left(1+{(30\epsilon(i-1))\cdot\rho^{-1}}\right) \cdot d_G(x,w_l)
	+ 6\sum_{j=1}^{i-1}R_j\cdot 2 ^{i-1-j} \\
	d_H(w_r,y) &\leq&
	\left(1+{(30\epsilon(i-1))\cdot\rho^{-1}}\right) \cdot d_G(w_r,y)
	+ 6\sum_{j=1}^{i-1}R_j\cdot 2 ^{i-1-j}.
\end{array}
\end{equation*}

	Therefore $d_H(x,y)$ is bounded by:

\begin{equation}\label{dhxy}
	\begin{array}{lllll}
	d_H(x,y)&\leq& 
	d_H(x,w_l)+d_H(w_l,r_{C_l})+d_H(r_{C_l},r_{C_r})+d_H(r_{C_r},w_r)+d_H(w_r,y)\\
	&\leq&
	d_H(x,w_l)+2R_i+1 +d_G(z_l,z_r)+2R_i +2R_i+1 +d_H(w_r,y)\\
	&\leq&
	\left(1+{(30\epsilon(i-1))}{\cdot\rho^{-1}}\right)d_H(x,y)	
	+ 6R_i+ 2\cdot 6\sum_{j=1}^{i-1}R_j\cdot 2 ^{i-1-j} \\
	&= &
	\left(1+{(30\epsilon(i-1))}{\cdot\rho^{-1}}\right)d_H(x,y)	
	+ 6\sum_{j=1}^{i}R_j\cdot 2 ^{i-j}.
	\end{array}
	\end{equation}

		\begin{figure}[h]
			\begin{center}
				\includegraphics[scale=0.6]{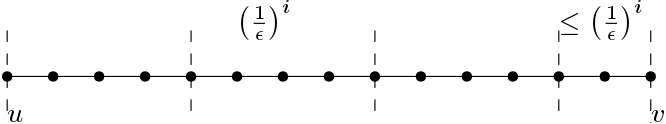}
				\caption{The path between $u,v$ in $G$ is divided into segments of length $\epsi$ (the last segment can be shorter).}
				\label{figure uv path}
			\end{center}
		\end{figure}
	
	Now, consider a pair $u,v\in V$ such that all vertices of $\pi(u,v)$ are $U^{(i)}$ clustered, but $\pi(u,v)$ may be of an arbitrarily large length. We divide the path into segments of length exactly $\epsi$, except for maybe one segment that can be shorter (see Figure \ref{figure uv path}).
	By \cref{dhxy}: 
	
\begin{equation}\label{eq55}
	\begin{array}{lllll}
	d_H(u,v) &\leq& \left(1+\frac{30\epsilon(i-1)}{\rho}\right)d_G(u,v) + \floor*{\frac{d_G(u,v)}{\epsi}}6\sum_{j=1}^{i}R_j\cdot 2 ^{i-j} +
	
	6\sum_{j=1}^{i}R_j\cdot 2 ^{i-j}\\

	&\leq&

	d_G(u,v)\left[
	\left(1+\frac{30\epsilon(i-1)}{\rho}\right)+
	6\epsilon^i\cdot \sum_{j=1}^{i}R_j\cdot 2 ^{i-j} 
	\right]+
	6\sum_{j=1}^{i}R_j\cdot 2 ^{i-j}.\\
	
	\end{array}
	\end{equation}

	It is left to show that 	
	$
	6{\epsilon^i}\cdot \sum_{j=1}^{i}R_j\cdot 2 ^{i-j} \leq
	\frac{30\epsilon}{\rho}
	$.
	
	First, we will provide an upper bound on $6\sum_{j=1}^{i}R_j\cdot 2 ^{i-j}$. Recall that by \cref{eq 4 rho}, $	R_{i}\leq ({4}/{\rho})\cdot\eps{(i-1)}$. 
%
%
	It follows that $
	6\sum_{j=1}^{i}R_j\cdot 2 ^{i-j} 
	\leq 
	6\sum_{j=1}^{i}({4}/{\rho})\cdot\eps{(j-1)}\cdot 2 ^{i-j} 
	=	
	{24}/({\rho\cdot \epsilon^{i-1}\cdot(1-2\epsilon)})$.
	Recall that $\epsilon<\frac{1}{10}$. Thus

\begin{equation}\label{bound 6sum}
	\begin{array}{cllllllllllll}
	6\sum_{j=1}^{i}R_j\cdot 2 ^{i-j} 
	
	&\leq &	
	\frac{24}{\rho\cdot \epsilon^{i-1}\cdot(1-2\epsilon)}
	
	&\leq&	
	\frac{30 }{\rho\cdot \epsilon^{i-1}}.\\
	
	\end{array}
	\end{equation}
	
	This yields: $
	6{\epsilon^i}\cdot \sum_{j=1}^{i}R_j\cdot 2 ^{i-j} 
	\leq {\epsilon^i}\cdot 	{30 }/({\rho\cdot \epsilon^{i-1}})
	=({30\epsilon})/{\rho}$. To conclude,

\begin{equation*}
	\begin{array}{lllll}
	d_H(u,v) 
	&\leq 
	d_G(u,v)\left[1+\frac{30\cdot \epsilon\cdot (i-1)}{\rho} + 
	
	\frac{30\epsilon}{\rho} \right]
	
	+ 6\sum\limits_{j=1}^{i}R_j 2 ^{i-j}
	
	&\leq
	d_G(u,v)\left(1+\frac{30\cdot\epsilon\cdot i}{\rho}\right) 
	
	+ 6\sum\limits_{j=1}^{i}R_j 2 ^{i-j}.
	\end{array}
	\end{equation*}

\end{proof}

Recall that by Corollary \ref{coro partition}, all vertices in $V$ are $U^{(\ell)}$-clustered, thus for any pair of vertices, the condition of Lemma \ref{lemma stretch} holds for $i=\ell$ and for every vertex $v\in V$ there is an index $i$ such that $v$ is $U_i$-clustered. By Lemma \ref{lemma stretch} and \cref{bound 6sum},

\begin{corollary}
	For every pair $u,v\in V$, for $\rho\geq 10\epsilon$, 
$
	d_H(u,v) \leq \left(1+\frac{30\cdot\epsilon\cdot \ell}{\rho}\right)\cdot d_G(u,v) 
	+ 
	\frac{30 }{\rho\cdot \epsilon^{\ell-1}}$.
\end{corollary}

\subsubsection{Rescaling}\label{sec rescale}

In this section, we rescale  $\epsilon$ to obtain a $(1+\epsilon,\beta)$-spanner. Set $\epsilon' = ({30\cdot\epsilon\cdot \ell})/{\rho}$. 

The condition $\epsilon< \frac{1}{10}$ now translates to $({\epsilon'\cdot \rho})/({30\ell})<{1}/{10}$, i.e, $\epsilon'<\frac{3\ell}{\rho}$. We replace it with a much stronger condition, $\epsilon'<1$. 
The condition $\rho\geq  10\epsilon$ translates to $\rho=({30\cdot\epsilon\cdot\ell})/{\epsilon'}> 10\epsilon$, i.e., $\epsilon'< 3\ell$, which holds trivially for $\epsilon'<1$ and $\ell\geq 1$.


It follows that the additive term of the stretch now translates to $
\frac{30 }{\rho\cdot \epsilon^{\ell-1}} 
\leq 	\left(
\frac{30\cdot \ell}{\rho\cdot \epsilon'}
\right)^{\ell}.
$

Denote $\beta = \left(
({30\cdot \ell})/({\rho\cdot \epsilon'})
\right)^{\ell}.
$
Observe that 
\begin{equation}
\label{eq eps beta}
\begin{array}{lcccccccc}
\beta 
&= &	\left(
({30\cdot \ell})/({\rho\cdot \epsilon'})
\right)^{\ell}
&=& \eps{\ell}.
\end{array}
\end{equation}

Thus we obtain stretch 
$
\left(
1+\epsilon', \left(
({30\cdot \ell})({\rho\cdot \epsilon'})
\right)^{\ell}
\right).
$

Recall that $\ell = \floor*{{\log \kappa\rho}}+ \lceil\frac{\kappa+1}{\kappa\rho}\rceil-1 \leq 
\log\kappa\rho +\rho^{-1}+O(1) $. It follows that:

\begin{equation*}
\beta = \left(
\frac{30\cdot \ell}{\rho\cdot \epsilon'}
\right)^{\ell}=
\left(
\frac{30\cdot (\log\kappa\rho +\rho^{-1}+O(1))}{\rho\cdot \epsilon'}
\right)^{\log\kappa\rho +\rho^{-1}+O(1)}.
\end{equation*}

By Corollary \ref{coro size} and \cref{eq eps beta}, the number of edges in the spanner $H$ is:
\begin{equation*}
\begin{array}{lclclclc}
|H|& = O(\nfrac\cdot \eps{\ell}) &= & O(\beta\cdot \nfrac).
\end{array}
\end{equation*}

%
%
%
%
%

By Corollary \ref{coro rt} and \cref{eq eps beta}, the time required to construct the spanner $H$ is: 
\begin{equation*}
\begin{array}{lclclclc}

O\left(n^\rho\cdot\rho^{-1}\eps{\ell}\right) = 
O\left(\beta\cdot n^\rho\cdot\rho^{-1}\right).

\end{array}
\end{equation*}

Denote now $\epsilon = \epsilon'$.
\begin{corollary}
	For any parameters $0<\epsilon\leq 1$, $\kappa\geq 2$, and ${1}/{\kappa}\leq \rho<{1}/{2}$, and any $n$-vertex graph $G=(V,E)$, our algorithm constructs a $(1+\epsilon,\beta)-spanner$ with
	$ O(\beta\cdot \nfrac) $
	edges,
	in 
	$ O\left(\beta\cdot n^\rho\cdot\rho^{-1}\right) $
	deterministic time in the CONGEST model, where

\begin{equation}
	\beta =
	\bbeta.
	\end{equation}

\end{corollary}

\bibliographystyle{alpha}
\bibliography{cite}

	\clearpage
\begin{appendix}

	\section{Detecting Popular Clusters} \label{section Detecting popular clusters}
	
	Given a graph $G=(V,E)$, a set of clusters $P_i$ centered around the vertices in $S_i$, parameters $\delta_i,deg_i$, we say that the cluster $C$ and its center $r_C$ are \textit{popular} if $|\Gamma^{(\delta_i)}(r_C)\cap S_i|\geq deg_i$, that is, if the center $r_C$ has at least $deg_i$ cluster centers within distance $\delta_i$ from it. This section provides a procedure that allows each cluster center $r_C$ for $C\in P_i$ to know if it is \textit{popular}.

	The procedure runs a modified BFS exploration from each vertex $r_C\in S_i$. In this exploration, 
	each vertex $u\in V$ maintains a list of all the cluster centers it learned about and the shortest known distance to them. The algorithm runs for $\delta_i$ phases. We note that the phases of this algorithm are different from the phases of the main algorithm described in Section \ref{sec construction of spanners}. Phase $0$ consists of a single round. Each phase $i>0$ consists of $deg_i$ rounds. Intuitively, the phases can be thought of as single rounds, in which vertices can send $deg_i$ messages.
	In round $0$ (also called phase $0$), each vertex $r_C\in S_i$ sends the message $\langle r_C,0\rangle$ to all its neighbors. The first part of the message is the ID of the original sender. The second part is the distance that the message has traversed so far. 
	In each phase $i>0$, each vertex $u\in V$ increments the distance on each message it received in the previous phase, and forwards these messages. Note that in each phase $j>0$, i.e., rounds $deg_i\cdot (j-1) +1$ to $deg_i\cdot j$ for $1\leq j\leq \delta_i$, each vertex sends messages that traversed exactly $j-1$ edges before reaching it.
	If a vertex $v\in V$ received messages of the form $\langle r_C,j\rangle$ regarding more than $deg_i$ centers, it will arbitrarily choose $deg_i$ of these messages and forward them. All other messages are discarded, i.e., they will never be sent.

	If a center $r_C\in S_i$ receives messages regarding at least $deg_i$ other clusters, it joins $A$. 
	Each vertex $u$ maintains a list of all the first $deg_i$ vertices it has learned about. The vertex $u$ maintains the shortest known distance to them, and its neighbor from which the message regarding each such center arrived. This is in order to trace back a shortest path from $u$ to vertices it has learned about.
	The pseudo-code of the algorithm (Algorithm \ref{Alg number of near neighbors}) is provided below.

	\begin{algorithm}[H]
		\caption{Number of near neighbors}
		\label{Alg number of near neighbors}
		\begin{algorithmic}[1]
			\State \textbf{Input:} graph $G=(V,E)$, a set of clusters $P_i$, parameters $deg_i,\delta_i$
			\State \textbf{Output:} a set $A$.
			\State Each vertex $v\in V$ initializes a list of centers it learned about and their distance from it as an empty list.
			\State Each $r_C\in S_i$ sends $\langle r_C.ID,0\rangle$ 
			\For {$j=1 \ to \ \delta_i$ }
			\For{$deg_i$ rounds}
			\If {$v$ received at most $deg_i$ messages $\langle r_C,j-1\rangle$}
			\State For each received messages $\langle r_C,j-1\rangle$, $v$ sends $\langle r_C,j\rangle$
			\EndIf
			\If {$v$ received more than $deg_i$ messages $\langle r_C,j-1\rangle$}
			\State {For arbitrary $deg_i$ received messages $\langle r_C,j-1\rangle$, $v$ sends $\langle r_C,j\rangle$}
			\EndIf
			\EndFor
			\EndFor 
			\State Each $r_C\in S_i$ that has learned about at least $deg_i$ other centers joins $A$. 
		\end{algorithmic}
	\end{algorithm}

	Next, we show that when the algorithm terminates, each vertex $u\in V$ maintains information regarding some of the centers in $S_i$ that are close to it. That is, $u$ knows the IDs of these vertices, its distance to them, and its neighbor that informed it of them.

	\begin{lemma}\label{lemma induc know neighbors}
		When Algorithm \ref{Alg number of near neighbors} terminates, each vertex $u\in V$ knows at least $\min\{deg_i,|\Gamma^{(\delta_i)}(u) \cap S_i| \}$ centers from $S_i$ with distance at most $\delta_i$ from $u$. 
	\end{lemma}
	
	\begin{proof}
		We will prove by induction on the index of the phase $j$, that by the end of phase $j$, for all $j\geq 0$, each vertex $u$ knows at least $\min\{deg_i,|\Gamma^{(j)}(u) \cap S_i| \}$ centers from $S_i$ with distance at most $j$ from $u$. 
		
		For $j=0$ the claim is trivial since $\Gamma^{(0)}(u) = \{u\}$.
		
		For $j>0$, assume inductively that by the end of phase $j-1$, each vertex $v\in V$ knows $\min\{deg_i,|\Gamma^{(j-1)}(v) \cap S_i| \}$ centers from $S_i$ with distance at most $j-1$ from it. 
		
		Let $j>0$ and let $u$ be a vertex. The immediate neighbors of $u$ in $G$ have exactly $|\Gamma^{(j)}(u) \cap S_i|$ distinct vertices from $S_i$ within distance at most $j-1$ from them. By the induction hypothesis, each vertex $v$ which is a neighbor of $u$ knows at least 
		$\min\{deg_i,|\Gamma^{(j-1)}(v) \cap S_i| \}$ centers from $S_i$ within distance at most $j-1$ from it. 
		
		If one of these neighbors has at least $deg_i$ vertices from $S_i$ within distance exactly $j-1$ from it, then it has sent $deg_i$ messages to $u$ in phase $j$. Thus by the end of phase $j$, $u$ knows at least $deg_i$ vertices from $S_i$ within distance at most $j$ from it and the claim holds. Otherwise, in phase $j$, each vertex $v$ which is a neighbor of $u$, sent to $u$ messages regarding all vertices in $S_i$ within distance exactly $j-1$ from $v$. Moreover, by the induction hypothesis, at the end of phase $j-1$, the vertex $u$ already knows 
		$\min\{deg_i,|\Gamma^{(j-1)}(u) \cap S_i| \}$ centers from $S_i$ with distance at most $j-1$ from $u$. It follows that by the end of phase $j$, the vertex $u$ knows $\min\{deg_i,|\Gamma^{(j)}(u) \cap S_i| \}$ centers from $S_i$ with distance at most $j$ from it.
	\end{proof}

	We are now ready to prove Theorem \ref{theorem popular}. For convenience, it is stated here again:
	
	\textit{\textbf{Theorem \ref{theorem popular}} 	Given a graph $G=(V,E)$, a collection of clusters $P_i$ centered around vertices $S_i$ and parameters $\delta_i,\ \deg_i$, Algorithm \ref{Alg number of near neighbors} returns a set $A$ in $O(deg_i\cdot\delta_i)$ time such that: 
		\begin{enumerate}
			\item $A$ is the set of all centers of popular clusters from $P_i$.
			\item Every cluster center $r_C\in S_i$ that did not join $A$ knows the identities of all the centers $r_{C'}\in S_i$ such that $d_G(r_C,r_{C'})\leq \delta_i$. Furthermore, for each pair of such centers $r_C,r_{C'}$, there is a shortest path $\pi$ between them such that all vertices on $\pi$ know their distance from $r_{C'}$.
		\end{enumerate} 
	}

	\begin{proof}

		(1) Let $C\in P_i$ be a popular cluster, i.e., $r_C$ has at least $deg_i$ other centers in $S_i$ within distance $\delta_i$ from it. By Lemma \ref{lemma induc know neighbors}, when the algorithm terminates, $r_C$ knows at least $deg_i$ other cluster centers from $S_i$. Thus it joins $A$. 
		
		Let $r_C\in A$. Observe that $r_C$ joined $A$ because it learned about at least $deg_i$ clusters within distance $\delta_i$ from it. Thus it is a popular cluster.
		
		(2) Let $r_C\in S_i\backslash A$. Since $r_C$ did not join $A$, it has learned about less than $deg_i$ other cluster centers from $S_i$. Therefore, by Lemma \ref{lemma induc know neighbors}, the vertex $u$ has learned about $\min\{deg_i,|\Gamma^{(\delta_i)}(r_C) \cap S_i| \} =|\Gamma^{(\delta_i)}(r_C) \cap S_i|$ vertices in $S_i$ with distance at most $\delta_i$ from it. Hence it knows the identities of all the centers $r_{C'}\in S_i$ such that $d_G(r_C,r_{C'})\leq \delta_i$.
		
		Let $r_C$ be a center that did not join $A$, and let $r_{C'}$ be a center $d_G(r_{C'},r_C)\leq \delta_i$. The messages regarding $r_{C'}$ reached $r_C$ along some shortest $(r_{C'},r_C)$ path $\pi$. Since each vertex $v\in V$ maintains a list with all (up to $deg_i$) the centers it has learned about and the distances to them, all the vertices on the shortest path $\pi$ know their distance to $r_{C'}$.

		The running time of Algorithm \ref{Alg number of near neighbors} is trivially $O(deg_i\cdot\delta_i)$, as it is executed for 
		$\delta_i$ phases, each of which lasts up to $deg_i$ rounds. 
	\end{proof}
	
	\clearpage
	\section{Previous Results}\label{sec append prev res}
	
	In this section, we review relevant previous results for near-additive spanners.

	\begin{table}[H]
		\centering
		\begin{adjustbox}{width=1\textwidth}
			
			\begin{tabular}{|l|l|l|l|l|}
				\hline
				\textbf{authors}&
				\textbf{model} &	 
				\textbf{stretch}	& 	
				\textbf{size} &	
				\textbf{running time}\\ 
				\hline

				\hline
				\cite{ElkinP01}	&		
				centralized, deterministic &	 
				$(1+\epsilon,4)$	& 	
				$O(\epsilon^{-1}n^{\frac{4}{3}})$ &	
				$O(mn^\frac{2}{3})$\\

				\hline
				\cite{ElkinP01}	&		
				centralized, deterministic &	 
				$\left(1+\epsilon,\beta\right)$	& 	
				$O\left(\beta \nfrac\right)$ &	
				$\tilde{O}(m\cdot n)$\\
				
				&		
				&	 
				$\beta = \left(\frac{{\log \kappa}}{\epsilon}\right)^{\log \kappa}$	& 	
				&	
				\\

				&		
				 &	 
				$\left(1+\epsilon,\beta'\right)$	& 	
				$O\left(\beta' \nfrac\right)$ &	
				$\tilde{O}(n^{2+\rho})$\\
				
				&		
				&	 
				$\beta' = {\max\{\left(\frac{{\log \kappa}}{\epsilon}\right)^{\log \kappa},\kappa^{-{\log \rho}} \}}$	& 	
				&	
				\\
				
				\hline
				\cite{Elkin05}	&
				\congest, deterministic&		
				$(1+\epsilon,\beta)$& 	
				$O(\nfrac)$&	
				$O(n^{1+\frac{1}{2\kappa}})$\\ 
				&		
				&	 
				$\beta= \left(\frac{\kappa}{\epsilon}\right)^{O({\log \kappa})}\cdot\frho^{\frho} $& &	
				\\ 
				&		
				&	 
				$\frac{1}{2\kappa}<\rho<\frac{1}{2\kappa}+\frac{1}{3}$& &	
				\\

				\hline
				\cite{ElkinZ06}	&		
				\congest, randomized&	 
				$(1+\epsilon,\beta)$, 
				& 	
				$O(\nfrac)$&	
				$O(n^\rho)$\\ 
				&		
				&	 
				$\beta= \left(\frac{\kappa}{\epsilon}\right)^{O({\log \kappa})}\cdot\frho^{\frho} $ Or  & 	
				&	\\ 
				&		
				&	 
				$\beta= \left(\frac{\kappa{\log \kappa}}{\epsilon}\right)^{O({\log \kappa})}\cdot\frho^{\frho} $   & 	
				&	
				\\ 
				&		
				&	 
				$\frac{1}{2\kappa}<\rho<\frac{1}{2\kappa}+\frac{1}{3}$& 	
				&	
				\\

				\hline
				\cite{ThorupZ06}	&		
				centralized, randomized&	 
				$\left( 1+\epsilon, \left(\frac{O(1)}{\epsilon}\right)^\kappa \right) $& 	
				
				$O(n^{1+\frac{1}{\kappa}})$&	
				$O(mn^{1/\kappa})$ 
				\\
				
				\hline
				\cite{DerbelGP07} 	& 
				\local, deterministic &	 
				$(1+\epsilon,8{\log n})$	& 	
				$O(n^\frac{3}{2})$ &	 
				$O(\frac{1}{\epsilon}{\log n})$\\
				
				\hline
				\cite{derbelGPV08} 	& 
				\local, deterministic &	 
				$(1+\epsilon,2)$	& 	
				$O(\epsilon^{-1}n^\frac{3}{2})$ &	 
				$O(\epsilon^{-1})$\\

				\hline
				\cite{DerbelGPV09}	&
				\local, deterministic		&
				$(1+\epsilon,O(\frac{1}{\epsilon})^{\kappa-2})$&	 
				$O(\epsilon^{-\kappa+1}\nfrac)$ & 	
				$O(1)$
				\\
				&
				&
				$(1+\epsilon,\beta)$, $\beta= \left(\frac{{\log \kappa}}{\epsilon}\right)^{O({\log \kappa})}$&	 
				$O(\beta\nfrac)$ & 	
				$ O(\beta\cdot 2^{O(\sqrt{{\log n}})}) $
				\\

				&
				&
				$(1+\epsilon,\beta)$, $\beta= 
				\left(
				\frac{{\log n}}{\epsilon}
				\right)^{O({\log {\log n}})}
				$&	 
				$O(\beta n)$ & 	
				$ O(\beta\cdot {\log n}) $
				\\

				\hline
				\cite{Pettie09} 	& 
				centralized, randomized &	 
				$\left( 1+\epsilon,O\left(  
				\epsilon^{-1}\cdot{\log{\log n}}
				\right)^{\log {\log n}}\right)$	& 	
				$O\left(n{\log{\log (\epsilon^{-1}{\log{\log n) }} }} \right)$ &	
				NA\\
				
				\hline
				\cite{Pettie09} 	& 
				centralized, randomized &	 
				$\left( 1+\epsilon,O\left(  
				\epsilon^{-1}\cdot{\log{\log n}}
				\right)^{\log {\log n}}\right)$	& 	
				$O\left((1+\epsilon)n \right)$ &	
				NA\\

				\hline
				\cite{Pettie10} 	& 
				\congest, randomized &	 
				$\left( 1+\epsilon,
				O\left(\frac{{\log \kappa}+\rho^{-1}}{\epsilon}\right)^
				{\log_{\phi}\kappa+\rho^{-1}} \right)$	& 	
				$O\left(\nfrac \left( \frac{{\log \kappa}}{\epsilon}\right)^\phi \right)$ &	
				$\tilde{O}(n^\rho)$\\
				& 
				&	 
				$\phi=\frac{1+\sqrt{5}}{2}$	& 	
				&	
				\\
				
				\hline
				\cite{AbboudBP17}&
				centralized, randomized &	 
				$(1+\epsilon,O(\frac{\log\kappa}{\epsilon})^{\log\kappa-1})$	& 	
				$O( (\frac{\log\kappa}{\epsilon})^{h_\kappa}\log\kappa n ^{1+\frac{1}{\kappa}})$, $h_\kappa<\frac{3}{4}$ &	
				NA\\
				
				\hline

				\cite{ElkinN17}
				&	\congest, randomized &	 
				$\left(1+\epsilon,
				\beta\right)$ 	& 	
				$O(n^{1+\frac{1}{\kappa}})$ &	$O(n^\rho\cdot\rho^{-1}\cdot \beta\cdot {\log n})$\\

				&	 &	 
				$\beta=
				O\left(
				\frac{{\log \kappa}+\rho^{-1}}
				{\epsilon}\right)^{\log(\kappa)+\rho^{-1}}$ 	& 	
				&
				\\ 
				
				\hline
				\hline
				\textbf{New}&
				\congest, deterministic &	 
				$\left(1+\epsilon,\beta\right)$	& 	
				$ O\left(\beta\cdot\nfrac\right)$ &	
				$O\left(\beta\cdot n^\rho\cdot\rho^{-1}\right)$\\

				&
				&	 
				$\beta=			 
				\bbeta $	& 	
				
				&	
				\\

				\hline
			\end{tabular}
		\end{adjustbox}
		\caption{Previous results for near-additive spanners} \label{table near-additive}	
	\end{table}

\end{appendix}

\end{document}